\documentclass[12pt]{amsart}

\usepackage[matrix,arrow,curve,frame]{xy}
\usepackage{amsmath,amsthm,amssymb,enumerate}
\usepackage{latexsym}
\usepackage{amscd}
\usepackage{xcolor}
\usepackage[backref]{hyperref}
\usepackage{euscript}
\usepackage{mathrsfs}
\usepackage{tikz}
\usepackage{tikz-cd}
\usepackage{comment}

\usepackage{fullpage}

\newtheorem{theorem}{Theorem}[section]
\newtheorem{lemma}[theorem]{Lemma}
\newtheorem{corollary}[theorem]{Corollary}
\newtheorem{proposition}[theorem]{Proposition}

\theoremstyle{definition}
\newtheorem{definition}[theorem]{Definition}

\newtheorem{remark}[theorem]{Remark}

\numberwithin{equation}{section}
\newcommand{\beq}{\begin{equation}}
\newcommand{\eeq}{\end{equation}}

\newcommand{\CC}{\mathbb{C}}

\newcommand{\sU}{{\sf U}}
\newcommand{\sV}{{\sf V}}

\DeclareMathOperator{\Tr}{Tr}

\newcommand {\be}{\begin{equation}}
\newcommand {\ee}{\end{equation}}
\newcommand{\h}{\begin{eqnarray*}}
\newcommand{\e}{\end{eqnarray*}}

\newtheorem*{theoremA}{Theorem A}

\newtheorem*{theoremB}{Theorem B}

\newtheorem*{theoremC}{Theorem C}

\begin{document}

\title[Brane space localization and genus-0 T-duality]
{Brane space localization and genus-0 T-duality with $H$-flux}

\author{Fei Han}
\address{Department of Mathematics,
National University of Singapore, Singapore 119076}
\email{mathanf@nus.edu.sg}

 \author{Varghese Mathai}
\address{School of Mathematical Sciences,
Adelaide University, Adelaide 5005, Australia}
\email{mathai.varghese@adelaide.edu.au}

\date{}

\begin{abstract}
We construct a geometric theory on the brane space
$S^2M=C^\infty(S^2,M)$ associated to a gerbe with connection on $M$.
Normalized logarithmic holonomy around the latitude circles defines an
average holonomy line bundle on the brane space, whose connection and averaged flux give an
equivariantly flat superconnection.  We prove that the completed periodic exotic twisted cohomology defined by this superconnection localizes to the twisted cohomology of $M$.
This {\bf brane space localization} happens in two stages: circle localization lands on the free path space of $M$, and contracting the paths, after a correction of the twisting form, brings us back to $M$.  

Parallel transport along the meridians lets us average gerbe module connections, after which the circular iterated-integral construction applies.  The resulting {\bf twisted
brane Bismut-Chern forms} restrict to the twisted Chern forms on $M$;
their variation is described by {\bf twisted brane Bismut-Chern-Simons forms.}

For a T-dual pair, transporting the classical Hori morphisms through localization yields {\bf graded genus-0 T-duality} on brane cohomology. The composite of each morphism with its dual is minus the Euler operator in the level variable, so these morphisms are isomorphisms at every nonzero level.  At level
one these morphisms are compatible with the brane character and
twisted $K$-theory T-duality.
\end{abstract}

\maketitle
\enlargethispage{2pt}

\tableofcontents

\section*{Introduction}
\subsection{Background}
T-duality is one of the more striking phenomena of string theory: it identifies theories whose target spaces can be geometrically quite different.  For a sigma model whose
target has a circle symmetry, it exchanges momentum around the circle with
winding along it and replaces a circle of radius $R$ by one of radius $R^{-1}$. T-duality for $N=(2,2)$ sigma models with H-flux and its relation to generalized Kähler geometry is studied in \cite{LRRUZ}; see also \cite{KL, HLSUZ} for the associated topological sigma models. The realisation is straightforward when the target space has an abelian isometry that is also a symmetry of the worldsheet action. One finds the transformation on the worldsheet fields as well
as Buscher's rules \cite{Buscher} which yield the dual background. The duality transformation is
symmetric in the sense that one can start from either theory (the original or the dual one)
and obtain the other one via the same gauging process \cite{Alvarez2}. A relevant account of the topological sigma
model is in \cite{Witten}.

 In the presence of a background
$H$-flux, the metric, the circle fibration, and the Wess-Zumino term belong to a
single piece of geometry. The Buscher rules transform them together, and the
topology of the target may change under the duality
\cite{Buscher,Alvarez2, BEM04a}.

There is a precise topological formulation of this principle. Let
\[
 \pi:Z\longrightarrow X,
 \qquad
 \widehat\pi:\widehat Z\longrightarrow X
\]
be principal circle bundles, and let $H$ and $\widehat H$ be closed integral
three-forms on their total spaces.  They form a T-dual pair when fibre
integration exchanges flux and Chern class,
\[
 \pi_*[H]=c_1(\widehat Z),
 \qquad
 \widehat\pi_*[\widehat H]=c_1(Z),
\]
and the two fluxes agree on the correspondence space in integral cohomology.
If $p:Z\times_X\widehat Z\to Z$ and
$\widehat p:Z\times_X\widehat Z\to\widehat Z$ are the projections, suitable
connections determine a 2-form $\mathcal F$ for which
\begin{equation}
 p^*H-\widehat p^*\widehat H=d\mathcal F.
 \label{eq:introduction-T-duality-identity}
\end{equation}
The existence and involutivity of the dual pair follow from the Gysin sequence
\cite{BEM04a,BEM04b,BunkeSchick}.

The differential form content of
\eqref{eq:introduction-T-duality-identity} is the Hori transform
\begin{equation}
 \omega\longmapsto
 \widehat p_*\bigl(e^{\mathcal F}p^*\omega\bigr),
 \label{eq:introduction-Hori-transform}
\end{equation}
which reverses parity and identifies the $H$-twisted cohomology of $Z$ with the
$\widehat H$-twisted cohomology of $\widehat Z$.  This displays in one formula the
exchange of topology and flux, and it places Ramond-Ramond fields and their
T-duals in the appropriate twisted cohomology groups
\cite{BEM04a,BCMMS,BM00, MR05}. These works were inspired by \cite{Hori, SYZ}.

\medskip

The sigma model, however, asks for more.  If the worldsheet is a closed surface
$\Sigma$, its bosonic configuration space is $C^\infty(\Sigma,Z)$, not $Z$.
Does T-duality itself have an infinite-dimensional form on the space of fields?  The program developed in \cite{HM15,HM18,HM22,HM24,HM25} approaches this question geometrically.  The flux on the target is first transgressed into coefficient
data on the mapping space.  One then seeks an equivariantly flat
superconnection in the sense of \cite{Q, MQ}, whose cohomology is accessible to localization, and finally
asks whether the Hori transform survives at this new level.  The guiding
principle comes from the fixed point ideas of Witten, Atiyah, Bismut, and
Jones-Petrack \cite{A85,B85,JP}: an infinite-dimensional theory may retain
finite-dimensional content, provided its circle symmetry and its twisting
data are treated together.

The first instance of the program concerns the one-dimensional sigma model.
For the loop space $LZ=C^\infty(S^1,Z)$, a gerbe on $Z$ has a holonomy line
bundle on $LZ$.  Its connection combines with the averaged flux and the vector
field which rotates loops to form an equivariantly flat superconnection.  The
resulting completed periodic exotic twisted equivariant cohomology localizes
to the twisted cohomology of $Z$, because the fixed loops are precisely the
constant loops.  In this way the classical T-duality isomorphism acquires a
loop-space realization, together with a twisted Bismut-Chern character
\cite{HM15}.  The level refinement and the loop Hori formulae reveal the
familiar exchange of momentum and winding in a graded geometric form
\cite{HM18,HM22,HM25}.

The second instance concerns the genus-1 two-dimensional sigma model.  The
configuration space is the double loop space $C^\infty(T^2,Z)$.  The two
translations of the torus, together with its complex structure parameter
$\tau$, lead to an average $\tau$-holonomy line bundle and to a superconnection
in the direction $K_1+\tau K_2$.  Localization again brings the theory to the
target, and the graded Hori morphisms preserve the Jacobi structure carried by
the fields \cite{HM21,HM24}.  These constructions suggest a common question: how do the geometry
and symmetry of the source determine the coefficient data and the
localization of the associated mapping-space theory?

\medskip

In this paper we take up the question for the sphere, that is, in genus-0. The sphere has no second translation direction, so its fixed point geometry is different, and the passage from that geometry back to the target lies at the heart of the results below.

\subsection{Main results} 
\subsubsection{Brane spaces and localization}
For a smooth manifold $M$ we write
\[
 S^2M=C^\infty(S^2,M)
\]
and call it the {\bf brane space of $M$}.  A choice of polar axis gives an action of
$\mathbb T_{\rm rot}$ on $S^2M$ by rotation of the source sphere.  The orbits on
$S^2$ are the circles of latitude.  If a gerbe on $M$ has curvature $H$, we
take, on each overlap, the holonomy of its line-bundle data around every
latitude.  The logarithm is normalized at the south pole, averaged from the
south pole to the north pole with the normalized spherical measure, and then
exponentiated.  The resulting transition functions define a complex line
bundle
\[
 \mathcal L_H\longrightarrow S^2M,
\]
which we call the {\bf average latitude holonomy line bundle}.  Topologically $\mathcal L_H$ is trivial; the essential geometry sits in its distinguished connection.  If $K$ denotes the infinitesimal rotation and
$\overline H$ the averaged 3-form, then
\begin{equation}
 \left(\nabla^{\mathcal L_H}\right)^2=\iota_K\overline H.
 \label{eq:introduction-average-curvature}
\end{equation}
Consequently
\begin{equation}
 Q_H=\nabla^{\mathcal L_H}-u\iota_K+u^{-1}\overline H,
 \qquad |u|=2,
 \label{eq:introduction-brane-superconnection}
\end{equation}
satisfies
\[
 Q_H^2+u\mathcal L_K^{\mathcal L_H}=0.
\]
On rotation-invariant forms, $Q_H$ therefore squares to zero.  We write
\[
 h^\bullet_{\mathbb T_{\rm rot}}
 (S^2M,\nabla^{\mathcal L_H}:\overline H)
 =H^\bullet\!\left(
 \Omega^\bullet(S^2M,\mathcal L_H)^{\mathbb T_{\rm rot}}
 [[u,u^{-1}]],Q_H\right)
\]
and call this the \textbf{completed periodic exotic twisted equivariant
cohomology of the brane space}.  The superscript on the space of forms
indicates invariance under source rotation.  The notation
$[[u,u^{-1}]]$ denotes the product over ordinary form degrees at each
fixed total degree, with $|u|=2$, as specified in
Section~\ref{subsec:brane-periodic-completion}.  Thus forms of
arbitrarily high ordinary degree are allowed within one total degree.  The coefficient line, its connection and the averaged flux all enter this theory.  This is the
{\em passage from target-space gerbe geometry to brane space geometry}.

The localization of this theory contains the main new phenomena. A sphere
which is fixed by axial rotation need not be constant.  It is constant along
each latitude, but it may vary from the south pole to the north pole.  Thus
\begin{equation}
 (S^2M)^{\mathbb T_{\rm rot}}\cong PM,
 \qquad
 PM=C^\infty([-1,1],M),
 \label{eq:introduction-fixed-path-space}
\end{equation}
rather than $(S^2M)^{\mathbb T_{\rm rot}}=M$.  This apparently small difference changes the shape of the proof.  We first use the exotic twisted Borel-Witten localization theorem for
an equivariantly flat superconnection with a possibly nonflat coefficient line
bundle. Restriction then identifies the brane theory with a
twisted theory on $PM$.  The contraction
\[
 (r_s\gamma)(z)=\gamma(sz)
\]
retracts $PM$ onto the constant paths.  The path-averaged flux is not literally
the pullback of $H$ under midpoint evaluation, but the two differ by an exact
form.  The corresponding exponential gauge transformation supplies the final
identification with the twisted cohomology of $M$.

We call this two-step mechanism {\bf brane space localization}:
\begin{equation}
 S^2M
 \xrightarrow{\text{circle localization}}PM
 \xrightarrow{\text{path space contraction}}M.
 \label{eq:introduction-brane-localization}
\end{equation}
The additional path contraction carries the twisting data from the
fixed point space to the target.  In particular, the
nonzero homotopy components of $S^2M$ contain no fixed points and disappear
after localization.  The surviving component still carries the full gerbe
connection through $\mathcal L_H$ before it is brought back to the target.

To state the result, let $j:PM\to S^2M$ be the fixed point inclusion,
let $c:M\to PM$ include the constant paths, and let $e_0:PM\to M$
be midpoint evaluation.  Write $H_P=j^*\overline H$.  The path
contraction determines a two-form $\eta$ with
\[
 d\eta=H_P-e_0^*H,\qquad c^*\eta=0.
\]
Multiplication by $e^{-u^{-1}\eta}$ gives the gauge isomorphism
$\Psi$ from the complex twisted by $e_0^*H$ to the one twisted by
$H_P$.  These maps enter the localization theorem as follows.
\begin{theoremA}[{\bf Brane localization theorem}, Theorem \ref{thm:brane-to-target}]
The chosen circle action and normalized measure determine an isomorphism
\begin{equation}
 \begin{split}
 \operatorname{Loc}:
 h^\bullet_{\mathbb T_{\rm rot}}
 (S^2M,\nabla^{\mathcal L_H}:\overline H)
 \longrightarrow
 H^\bullet\Bigl(
 \Omega^\bullet(M)[[u,u^{-1}]],
 d+u^{-1}H
 \Bigr),
 \end{split}
 \label{eq:introduction-brane-to-target-isomorphism}
\end{equation}
given by
\begin{equation}
 \operatorname{Loc}
 =c^*\circ\Psi^{-1}\circ j^*.
 \label{eq:introduction-localization-map}
\end{equation}
Because $c^*\eta=0$, $\operatorname{Loc}$ is also induced by
restriction to constant spheres.  The factorization records the
geometric argument proving that this restriction is an isomorphism.

The construction is summarized by the commutative diagram
\begin{equation}
\begin{CD}
h^\bullet_{\mathbb T_{\rm rot}}(S^2M,\nabla^{\mathcal L_H}:\overline H)
@>{j^*}>>
H^\bullet(PM,d+u^{-1}H_P)\\
@V{\operatorname{Loc}}V{\cong}V
@V{c^*\circ\Psi^{-1}}V{\cong}V\\
H^\bullet(M,d+u^{-1}H)
@=
H^\bullet(M,d+u^{-1}H).
\end{CD}
 \label{eq:introduction-localization-diagram}
\end{equation}
\end{theoremA}

\subsubsection{Twisted brane Bismut-Chern character}

The localization theorem suggests a corresponding question about
characteristic forms.  In ordinary Chern-Weil theory, the curvature of
a vector bundle connection represents its Chern character.  In the
presence of a gerbe, one replaces vector bundles by gerbe modules and
combines their curvature with the local curving of the gerbe.  The
resulting twisted Chern character takes values in twisted cohomology
\cite{BCMMS,MS}. Let $\mathcal G$ be the chosen gerbe on $M$, with curvature $H$.
We write $K^\bullet(M,H)$ for complex $K$-theory twisted by the
integral Dixmier-Douady class of $\mathcal G$.  This integral class,
including its torsion, is understood as part of the notation.  An even class is represented by a pair
$(E,E')$ of gerbe modules.  Admissible connections on this pair give
a periodic twisted Chern form
$\operatorname{Ch}_H(\nabla^E,\nabla^{E'})$, obtained locally by
multiplying the traced difference of their curvature exponentials by
the curving exponential.  For a general twist the modules are Hilbert
bundles, and admissibility is the relative trace-class condition
recalled in Section \ref{subsec:character-review}.  This is exactly what is needed for the paired trace to make sense.  The resulting form satisfies
\[
 (d+u^{-1}H)\operatorname{Ch}_H(\nabla^E,\nabla^{E'})=0,
\]
and its cohomology class depends only on $[E,E']$.  Together with
suspension in odd degree, this defines the twisted Chern character
\[
 \operatorname{Ch}_H:K^\bullet(M,H)\longrightarrow
 H^\bullet\bigl(\Omega^\bullet(M)[[u,u^{-1}]],d+u^{-1}H\bigr).
\]

The loop-space refinement of this construction has its origin in
index theory.  For a closed even-dimensional spin manifold, Atiyah, developing an idea
of Witten, interpreted the index of the spin Dirac operator formally
as an integral of an equivariantly closed differential form over loop space.
Formal localization at the constant loops then gives the
Atiyah-Singer index formula \cite{A85}.  Bismut extended this picture
to a Dirac operator coupled to a vector bundle with connection
\cite{B85}.  The Bismut-Chern form supplies the bundle contribution
before localization.  It is constructed from iterated integrals of
curvature inserted into circular transport, and its
restriction to constant loops is the ordinary Chern form.  

In \cite{HM15} we constructed the corresponding twisted
Bismut-Chern character.  Let $LM=C^\infty(S^1,M)$, let $K_L$
generate loop rotation, and let
$\mathcal L_H^{\mathrm{loop}}\to LM$ be the holonomy line bundle
of the gerbe $\mathcal G$.  For $\theta\in S^1$, evaluation is the map
\[
 \operatorname{ev}_\theta:LM\longrightarrow M,
 \qquad \gamma\longmapsto\gamma(\theta).
\]
With normalized circle measure, put
\[
 \overline H_L=\int_{S^1}\operatorname{ev}_\theta^*H\,d\theta,
 \qquad
 D_H^L=\nabla^{\mathcal L_H^{\mathrm{loop}}}
       -u\iota_{K_L}+u^{-1}\overline H_L.
\]
The essential point is that the curvature of the holonomy connection
is tied to the averaged flux.  With our sign conventions, one has
\[
 \left(\nabla^{\mathcal L_H^{\mathrm{loop}}}\right)^2
 =\iota_{K_L}\overline H_L,
 \qquad
 (D_H^L)^2+u\mathcal L_{K_L}^{\mathcal L_H^{\mathrm{loop}}}=0.
\]
Thus $D_H^L$ is an equivariantly flat superconnection.  The flux term cancels the curvature, and this is what allows the holonomy line bundle, which is in general not flat, to serve as the coefficient system of a cohomology theory.
This is the key geometric feature of the exotic theory of \cite{HM15}: it retains the gerbe connection on loop space and at the
same time admits a fixed point localization theorem.

More precisely, the completed periodic exotic twisted equivariant
cohomology of loop space is
\[
h^\bullet_{\mathbb T_{\rm rot}}
 (LM,\nabla^{\mathcal L_H^{\mathrm{loop}}}:\overline H_L):=H^\bullet\!\left(
 \Omega^\bullet(LM,\mathcal L_H^{\mathrm{loop}})^{\mathbb T_{\rm rot}}
 [[u,u^{-1}]],D_H^L\right).
\]
Here the same circle notation and periodic completion are used as in
the brane theory.  The twisted Bismut-Chern form takes values in this
holonomy line bundle and is closed for $D_H^L$, so it defines a class
in the displayed group.  Over the constant loops, the line bundle has
a canonical trivialization and $D_H^L$ restricts to $d+u^{-1}H$.
Restriction therefore gives the localization map
$\operatorname{res}_L$.  Its compatibility with the characteristic
forms is expressed by the commutative diagram
\begin{equation}
 \begin{tikzcd}[column sep=large,row sep=large]
 K^\bullet(M,H)
 \arrow[r,"\operatorname{BCh}_H"]
 \arrow[dr,"\operatorname{Ch}_H"'] &
 h^\bullet_{\mathbb T_{\rm rot}}
 (LM,\nabla^{\mathcal L_H^{\mathrm{loop}}}:\overline H_L)
 \arrow[d,"\operatorname{res}_L","\cong"']\\
 &H^\bullet\bigl(\Omega^\bullet(M)[[u,u^{-1}]],d+u^{-1}H\bigr).
 \end{tikzcd}
 \label{eq:introduction-loop-character-diagram}
\end{equation}

In this paper, we seek a genus-0 counterpart of this diagram, with the brane space
$S^2M$ in place of $LM$.  The corresponding characteristic form will
be called the {\bf twisted brane Bismut-Chern character}.  Its
construction combines the parallel transport geometry of the loop
character with the average latitude holonomy line and the two-step
brane localization of Theorem~A.

At first sight, the latitude circles seem to give an immediate way of
extending this construction to $S^2M$. Each latitude defines a map
$r_t:S^2M\to LM$, so one can pull back the twisted Bismut-Chern form
along $r_t$.  These forms, however, have coefficients in the different
line bundles $r_t^*\mathcal L_H^{\mathrm{loop}}$.  Averaging these forms has no intrinsic meaning until the fibres have been compared. Moreover, the transition law we need on the brane space is an exponential average of logarithmic holonomies, and the construction must reproduce exactly this line bundle, together with its connection.

We resolve this difficulty by viewing a sphere in the other order.
Set
\[
 \mathscr P M=C^\infty([-\pi/2,\pi/2],M),
\]
the latitude path space with free endpoints.  For fixed longitude,
traversing the latitudes traces a meridian path.  Varying the longitude
then gives a loop of such paths, expressed by the map
\[
 \Gamma:S^2M\longrightarrow L(\mathscr P M),
 \qquad\Gamma(x)(\theta)(t)=x(c_t(\theta)).
\]
Parallel transport along the meridians brings the gerbe and module
connections to their fibres at the initial endpoint.  Their affine
averages define a gerbe with connections on $\mathscr P M$.
We prove that its loop holonomy line bundle pulls back under $\Gamma$ to
$\mathcal L_H$, preserving the distinguished connections.  The
twisted loop-space character can therefore be applied to this averaged
data and pulled back to the brane space.  It is this comparison of coefficient lines that makes the whole construction possible.

Let $\iota_M:M\to S^2M$ be the inclusion of constant sphere maps. Our second main result is the following.

\begin{theoremB}[{\bf Twisted brane Bismut-Chern character},
Theorem~\ref{thm:character-localization}]
There is a natural homomorphism
\begin{equation}
 \operatorname{BBCh}_H:K^\bullet(M,H)
 \longrightarrow h^\bullet_{\mathbb T_{\rm rot}}
 (S^2M,\nabla^{\mathcal L_H}:\overline H),
 \label{eq:introduction-brane-character}
\end{equation}
represented by explicit {\bf twisted brane Bismut-Chern forms}.  For an
admissible pair of gerbe modules with connection, these forms satisfy
\begin{equation}
 \begin{split}
 Q_H\operatorname{BBCh}_H(\nabla^E,\nabla^{E'})&=0,\\
 \iota_M^*\operatorname{BBCh}_H(\nabla^E,\nabla^{E'})
 &=\operatorname{Ch}_H(\nabla^E,\nabla^{E'}).
 \end{split}
 \label{eq:introduction-character-restriction}
\end{equation}
Changes of connections are governed by the {\bf twisted brane
Bismut-Chern-Simons form} of Definition \ref{def:character-brane-cs}.
On cohomology, one has
\begin{equation}
 \operatorname{Loc}\circ\operatorname{BBCh}_H
 =\operatorname{Ch}_H.
 \label{eq:introduction-character-localization}
\end{equation}
Equivalently, the following diagram commutes:
\begin{equation}
 \begin{tikzcd}[column sep=large,row sep=large]
 K^\bullet(M,H)
 \arrow[r,"\operatorname{BBCh}_H"]
 \arrow[dr,"\operatorname{Ch}_H"'] &
 h^\bullet_{\mathbb T_{\rm rot}}
 (S^2M,\nabla^{\mathcal L_H}:\overline H)
 \arrow[d,"\operatorname{Loc}","\cong"']\\
 &H^\bullet\bigl(\Omega^\bullet(M)[[u,u^{-1}]],d+u^{-1}H\bigr).
 \end{tikzcd}
 \label{eq:introduction-character-diagram}
\end{equation}
\end{theoremB}
\begin{remark}Replacing the gerbe by its tensor powers gives the corresponding characters, and the diagrams above, at every integral flux level.
\end{remark}

\subsubsection{Genus-0 T-duality}
We turn now to T-duality.  Let $(Z,H)$ and
$(\widehat Z,\widehat H)$ be a T-dual pair of principal circle bundles.  We use
$\mathbb T$ and $\widehat{\mathbb T}$ for their structure circles and reserve
$\mathbb T_{\rm rot}$ for rotation of the source sphere.  The structure circles
act on the corresponding brane spaces by postcomposition, and these actions
commute with $\mathbb T_{\rm rot}$.  Put
\begin{equation}
 R=\mathbb C[u,u^{-1}],
 \qquad |u|=2.
 \label{eq:introduction-periodic-ring}
\end{equation}
The completed tensor products below are taken in ordinary form
degree at each fixed total degree.
For every $m\in\mathbb Z$, the level $m$ superconnections on the two brane
spaces are
\begin{equation}
 \begin{split}
 Q_{mH}&=\nabla^{\mathcal L_H^{\otimes m}}
 -u\iota_K+u^{-1}m\overline H,\\
 Q_{m\widehat H}&=\nabla^{\mathcal L_{\widehat H}^{\otimes m}}
 -u\iota_{\widehat K}+u^{-1}m\overline{\widehat H},
 \end{split}
 \label{eq:introduction-level-superconnections}
\end{equation}
where $K$ and $\widehat K$ are generated by the same rotation of the source.
On the finite dimensional targets we use the periodic twisted differentials
\begin{equation}
 D^u_{mH}=d+u^{-1}mH\wedge,
 \qquad
 D^u_{m\widehat H}=d+u^{-1}m\widehat H\wedge.
 \label{eq:introduction-periodic-differentials}
\end{equation}

Brane space localization gives, at each level, the two isomorphisms
\begin{equation}
 \operatorname{Loc}_{Z,m}:\quad
 H^\bullet\!\left(
 \Omega^\bullet(S^2Z,\mathcal L_H^{\otimes m})^{
 \mathbb T_{\rm rot}\times\mathbb T}\widehat\otimes R,Q_{mH}
 \right)\xrightarrow{\ \cong\ }
 H^\bullet\!\left(
 \Omega^\bullet(Z)^{\mathbb T}\widehat\otimes R,D^u_{mH}
 \right),
 \label{eq:introduction-level-localization}
\end{equation}
and the analogous map $\operatorname{Loc}_{\widehat Z,m}$ on the dual side.
Thus the two infinite-dimensional theories which enter the main theorem are
the graded spaces
\begin{equation}
 \begin{split}
 \mathcal E_{S^2}(Z,H)
 &=\bigoplus_{m\in\mathbb Z}
 H^\bullet\!\left(
 \Omega^\bullet(S^2Z,\mathcal L_H^{\otimes m})^{
 \mathbb T_{\rm rot}\times\mathbb T}\widehat\otimes R,Q_{mH}
 \right)y^m,\\
 \mathcal E_{S^2}(\widehat Z,\widehat H)
 &=\bigoplus_{m\in\mathbb Z}
 H^\bullet\!\left(
 \Omega^\bullet(S^2\widehat Z,
 \mathcal L_{\widehat H}^{\otimes m})^{
 \mathbb T_{\rm rot}\times\widehat{\mathbb T}}\widehat\otimes R,
 Q_{m\widehat H}
 \right)y^m.
 \end{split}
 \label{eq:introduction-graded-brane-theories}
\end{equation}

For comparison, the two finite dimensional graded periodic theories are
\begin{equation}
 \begin{split}
 \mathcal P(Z,H)
 &=\bigoplus_{m\in\mathbb Z}
 H^\bullet\!\left(
 \Omega^\bullet(Z)^{\mathbb T}\widehat\otimes R,D^u_{mH}
 \right)y^m,\\
 \mathcal P(\widehat Z,\widehat H)
 &=\bigoplus_{m\in\mathbb Z}
 H^\bullet\!\left(
 \Omega^\bullet(\widehat Z)^{\widehat{\mathbb T}}\widehat\otimes R,
 D^u_{m\widehat H}
 \right)y^m.
 \end{split}
 \label{eq:introduction-graded-periodic-theories}
\end{equation}
The variable $y$ records the tensor power of the average holonomy line bundle,
or equivalently the flux level.

Let $T_{m,u}$ be the periodic level $m$ Hori morphism on the target.  Localization gives the cohomological genus-0 Hori morphism by
\begin{equation}
 T_m^{S^2}
 =\operatorname{Loc}_{\widehat Z,m}^{-1}\circ
 T_{m,u}\circ\operatorname{Loc}_{Z,m}.
 \label{eq:introduction-level-brane-Hori}
\end{equation}
The maps $T_m^{S^2}$ assemble into the graded genus-0 Hori morphism
\begin{equation}
 \operatorname{GHor}_{S^2}:
 \mathcal E_{S^2}(Z,H)\longrightarrow
 \mathcal E_{S^2}(\widehat Z,\widehat H),
 \qquad
 \sum_m\alpha_my^m\longmapsto
 \sum_mT_m^{S^2}(\alpha_m)y^m.
 \label{eq:introduction-graded-brane-Hori}
\end{equation}
If $\operatorname{GHor}_u$ denotes the graded periodic Hori morphism on the
targets and
\[
 \operatorname{Loc}_Z=\bigoplus_m\operatorname{Loc}_{Z,m},
 \qquad
 \operatorname{Loc}_{\widehat Z}
 =\bigoplus_m\operatorname{Loc}_{\widehat Z,m},
\]
then the definition gives
\begin{equation}
 \operatorname{GHor}_{S^2}
 =\operatorname{Loc}_{\widehat Z}^{-1}\circ
 \operatorname{GHor}_u\circ\operatorname{Loc}_Z.
 \label{eq:introduction-graded-transport}
\end{equation}
The commutative diagram in the following theorem places this identity
in its geometric setting.  Its upper arrow is T-duality on the brane
spaces, and its lower arrow is the graded Hori morphism on the targets.
The vertical arrows are supplied by the genus-0 localization
theorem: restriction to the fixed path spaces is followed by path
contraction with its correction of the twisting form. In short, the Hori identities are imported from the targets, whereas the coefficient connections, the localization and the characteristic forms all live on the brane spaces themselves.

\begin{theoremC}[{\bf Graded genus-0 T-duality with $H$-flux}, Theorem~\ref{thm:graded-genus-zero-T-duality}]
The graded genus-0 Hori morphism and its dual are parity reversing and
satisfy
\begin{equation}
 \widehat{\operatorname{GHor}}_{S^2}\circ
 \operatorname{GHor}_{S^2}
 =-y\frac{\partial}{\partial y},
 \qquad
 \operatorname{GHor}_{S^2}\circ
 \widehat{\operatorname{GHor}}_{S^2}
 =-y\frac{\partial}{\partial y}.
 \label{eq:introduction-graded-Euler-relations}
\end{equation}
These morphisms fit into the commutative diagram
\begin{equation}
\begin{CD}
\mathcal E_{S^2}(Z,H)
@>{\operatorname{GHor}_{S^2}}>>
\mathcal E_{S^2}(\widehat Z,\widehat H)\\
@V{\operatorname{Loc}_Z}V{\cong}V
@V{\operatorname{Loc}_{\widehat Z}}V{\cong}V\\
\mathcal P(Z,H)
@>{\operatorname{GHor}_u}>>
\mathcal P(\widehat Z,\widehat H).
\end{CD}
 \label{eq:introduction-graded-brane-Hori-diagram}
\end{equation}
Consequently, the graded Hori morphism is an isomorphism on the sum
of the nonzero levels.  At level one it gives the parity-reversing
T-duality isomorphism on the exotic cohomology of the brane spaces.

\end{theoremC}

The characteristic forms are compatible with this duality. Let
\[
 T_K:K^j(Z,H)\longrightarrow K^{j+1}(\widehat Z,\widehat H),
 \qquad j\in\mathbb Z/2\mathbb Z,
\]
be the classical twisted $K$-theory T-duality isomorphism
\cite{BEM04a,BunkeSchick}.  It is obtained by pullback to the
correspondence space, the chosen identification of twists there,
and the oriented circle pushforward.
Averaging the brane forms over the structure circles gives invariant
character classes, denoted by $\operatorname{BBCh}^{\mathbb T}_H$
and $\operatorname{BBCh}^{\widehat{\mathbb T}}_{\widehat H}$.
Theorem \ref{thm:character-t-duality} proves that
\begin{equation}
 T^{S^2}_1\operatorname{BBCh}^{\mathbb T}_H
 =\operatorname{BBCh}^{\widehat{\mathbb T}}_{\widehat H}T_K.
 \label{eq:introduction-character-t-duality}
\end{equation}

Together with the restriction theorem, this gives a commutative diagram
relating all three theories.  We write
$h^j_{\mathbb T_{\rm rot}}(S^2Z,Q_H)^{\mathbb T}$ and
$H^j(Z,D^u_H)^{\mathbb T}$ for the two cohomology groups in
\eqref{eq:introduction-level-localization} at $m=1$, and use the
corresponding notation on the dual side.  The superscripts indicate
that the complexes are formed from differential forms invariant under
the structure circles.  All four groups use the form-degree completion
over $R$, and the equivariant parameter $u$ continues to refer
to rotation of the source.  With this notation, the diagram is
\begingroup
\small
\begin{equation}
\begin{tikzcd}[column sep=2.4em,row sep=3.5em]
 K^j(Z,H)
 \arrow[r,"T_K","\cong"']
 \arrow[d,"\operatorname{BBCh}^{\mathbb T}_H"']
 \arrow[dd,out=180,in=180,looseness=2.05,
        "\operatorname{Ch}_H"']
 &K^{j+1}(\widehat Z,\widehat H)
 \arrow[d,"\operatorname{BBCh}^{\widehat{\mathbb T}}_{\widehat H}"]
 \arrow[dd,out=0,in=0,looseness=2.05,
        "\operatorname{Ch}_{\widehat H}"]\\
 h^j_{\mathbb T_{\rm rot}}(S^2Z,Q_H)^{\mathbb T}
 \arrow[r,"T^{S^2}_1","\cong"']
 \arrow[d,"\operatorname{Loc}_{Z,1}"',"\cong"]
 &h^{j+1}_{\mathbb T_{\rm rot}}
 (S^2\widehat Z,Q_{\widehat H})^{\widehat{\mathbb T}}
 \arrow[d,"\operatorname{Loc}_{\widehat Z,1}","\cong"']\\
 H^j(Z,D^u_H)^{\mathbb T}
 \arrow[r,"T_{1,u}","\cong"']
 &H^{j+1}(\widehat Z,D^u_{\widehat H})^{\widehat{\mathbb T}}.
\end{tikzcd}
 \label{eq:introduction-enhanced-character-duality}
\end{equation}
\endgroup
The outer curved arrows are the twisted Chern characters on the
targets.  Each factors through the characteristic forms on the brane
space and the localization isomorphism.  The upper square expresses
compatibility with twisted $K$-theory T-duality, while the lower square
is the level-1 Hori construction.

\medskip
The space $S^2M$ has another mathematical life.  In algebraic topology, the
study of operations on sphere mapping spaces is known as {\em higher string
topology or brane topology}.  Brane products, coproducts, Thom-spectrum models,
and actions of higher little-disks operads have been constructed in this
setting \cite{CV,Hu,GTZbrane,Wakatsuki}.  Higher Hochschild methods also give
models for mapping spaces and for surface products across all genera
\cite{GTZsurface}. We come to the same sphere brane from a different direction.  Its basic objects are differential forms with gerbe
twisted coefficients, its symmetry is the axial circle action, and its central
operation is localization rather than intersection and gluing. The proximity
of the two theories is nevertheless suggestive. It is natural to ask whether
the average holonomy line bundle admits multiplicative refinements compatible
with brane products, whether the superconnection has a higher Hochschild
model, and whether genus-0 Hori morphisms intertwine the resulting brane
operations.  These questions would connect T-duality to the algebraic
structures already present on sphere spaces.

\medskip

The genus-0 and genus-1 theories also point towards the higher-genus
problem.  They show that the correct infinite-dimensional theory is sensitive
not merely to the dimension of the source, but to its symmetry and its
decomposition.  A sphere has an axial circle action whose quotient is an
interval.  A torus has translations and a complex modulus.  A surface of
genus greater than one has no comparable connected symmetry in general, so a
single fixed point argument cannot be expected to suffice.  One is led instead
to seek a theory assembled from cutting and gluing, compatible with pants
decompositions, mapping class group actions, and the higher Hochschild or
factorization structures of surface mapping spaces.  The genus-0
construction supplies a new local model for such a theory: circle localization
may end on a mapping space with boundary data, and a further geometric
contraction may still recover the target.  Understanding how these local models behave under sewing is, in our view, the natural next step in the program.

\subsection*{Organization of the paper}

Section \ref{sec:brane-space} constructs the average latitude holonomy line
bundle $\mathcal L_H\to S^2M$ from gerbe data on $M$.  We calculate its
connection and curvature, introduce the averaged flux, and prove the
equivariant flatness of the superconnection $Q_H$.  This gives the completed
periodic exotic twisted cohomology of the brane space.

Section \ref{sec:brane-localization} develops the localization argument needed
for this theory.  We first recast the exotic twisted Borel-Witten fixed point
theorem and then identify the fixed point set of $S^2M$ with $PM$
and contract $PM$ to $M$, with an explicit gauge correction for the twisting
3-form.  The outcome is the brane space localization isomorphism of
Theorem \ref{thm:brane-to-target}.

Section \ref{sec:brane-character} begins with a review of twisted
Chern characters and their loop-space refinements.  We then form the
averaged gerbe and its module connections on the latitude path space,
and identify the pullback of its loop holonomy line with $\mathcal L_H$.
This gives the twisted brane Bismut-Chern and brane Bismut-Chern-Simons forms.
We prove their restriction identities, treat both parities, and establish
compatibility with localization in Theorem \ref{thm:character-localization}.

Section \ref{sec:genus-zero-t-duality} reviews T-duality with $H$-flux and
the classical Hori transform, then passes to integral levels and
completed periodic cohomology.  Transport through brane localization
gives the graded genus-0 Hori morphisms and the Euler identities of
Theorem \ref{thm:graded-genus-zero-T-duality}.  The final subsection
proves their level-1 compatibility with the brane character and
twisted $K$-theory T-duality.

\subsection*{Acknowledgement}  Fei Han was partially supported by the AcRF grant from the National University of Singapore.  Varghese Mathai thanks the Australian Research Council  for support via the Australian Laureate Fellowship FL170100020.

\section{Average holonomy line bundles and exotic twisted theory on brane spaces}
\label{sec:brane-space}

The purpose of this section is to pass from gerbe data on a smooth manifold $M$
to geometric and cohomological data on its brane space $S^2M$. The construction
rests on a simple feature of the sphere.  Once a polar axis has been chosen, $S^2$
is swept out by the orbits of a circle action, namely its circles of latitude.  We
first use these orbits to average differential forms and to define a latitude
transgression.  We then apply the same procedure to the local holonomy of a gerbe
and obtain a line bundle with connection on the brane space $S^2M$.  The curvature of this
connection is precisely the term needed to construct an equivariantly flat
superconnection.  Its cohomology is the exotic twisted theory which will be used
in the sequel.

\subsection{The brane space and its circle action}

Let $M$ be a smooth manifold.  We write
\begin{equation}
 S^2M=C^\infty(S^2,M)
 \label{eq:brane-space}
\end{equation}
for the {\bf brane space of $M$}.  This is a Fr\'echet manifold.  At a map
$x\in S^2M$, its tangent space is naturally identified with the space of smooth
sections of the pullback tangent bundle:
\begin{equation}
 T_xS^2M=\Gamma(S^2,x^*TM).
 \label{eq:tangent-brane-space}
\end{equation}

Choose spherical coordinates
\[
 (\theta,t)\in\mathbb T\times[-\pi/2,\pi/2],
 \qquad \mathbb T=\mathbb R/\mathbb Z,
\]
where $t$ is the latitude and $\theta$ parametrizes the corresponding circle of
latitude.  At $t=-\pi/2$ and $t=\pi/2$, the circle collapses to the south and north
poles, respectively.  We denote the latitude at height $t$ by
\[
 c_t:\mathbb T\longrightarrow S^2.
\]

The group $\mathbb T$ acts on $S^2$ by rotation about the polar axis and preserves the
symplectic form on $S^2$.  Our
orientation convention is
\[
 a\cdot(\theta,t)=(\theta+a,t),
 \qquad a\in\mathbb T,
\]
and we write $V$ for the corresponding infinitesimal generator.  By precomposition,
the same circle acts on $S^2M$.  Thus
\[
 (a\cdot x)(p)=x(a\cdot p),
\]
and the induced vector field $K$ on $S^2M$ is given by
\begin{equation}
 K_x=dx(V)\in\Gamma(S^2,x^*TM).
 \label{eq:brane-rotation-vector-field}
\end{equation}
The convention is illustrated in Figure \ref{fig:sphere-circle-action}.  Both
arrows in the figure point in the positive direction determined by $V$.

\begin{figure}[ht]
\centering
\begin{tikzpicture}[scale=1.05,line cap=round,line join=round]
 \def\R{2.12}
 \draw[black,thick] (0,0) circle (\R);
 \draw[black!45,densely dashed] (0,-2.48)--(0,2.48);
 \node[below] at (0,2.08) {$N$};
 \node[below] at (0,-2.15) {$S$};

 \foreach \y/\a in {-1.35/1.64,-0.70/2.00,0.72/1.99,1.36/1.63}{
  \draw[black!35,densely dashed]
   (-\a,\y) arc[
      start angle=180,
      end angle=360,
      x radius=\a,
      y radius=.18
   ];
  \draw[black!55]
   (\a,\y) arc[
      start angle=0,
      end angle=180,
      x radius=\a,
      y radius=.18
   ];
 }

 \def\yt{0.08}
 \def\at{2.12}

 \draw[black!55,thick,densely dashed]
  (-\at,\yt) arc[
     start angle=180,
     end angle=360,
     x radius=\at,
     y radius=.27
  ];

 \draw[black,very thick,->]
  (\at,\yt) arc[
     start angle=0,
     end angle=152,
     x radius=\at,
     y radius=.27
  ];

 \node[black,right] at (2.20,.35) {$c_t(\mathbb T)$};
 \node[black] at (-1.10,.56) {$V$};

 \draw[black,thick,->]
  (.92,2.62) arc[
     start angle=0,
     end angle=152,
     x radius=.92,
     y radius=.22
  ];

 \node[black,anchor=west] at (1.12,2.62)
  {$S^1\simeq\mathbb T\curvearrowright S^2$};
\end{tikzpicture}
\caption{The latitude orbits of the circle action on $S^2$.  The vector field
$V$ generates rotation about the polar axis, and $c_t(\mathbb T)$ is one of its
orbits.  Both arrows represent the positive rotation
$\theta\mapsto\theta+a$.}
\label{fig:sphere-circle-action}
\end{figure}
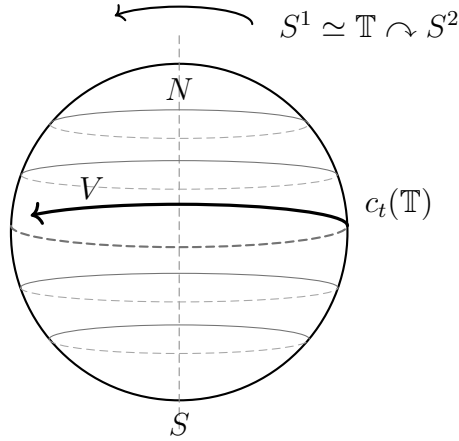

Let
\begin{equation}
 \operatorname{ev}:S^2M\times S^2\longrightarrow M,
 \qquad
 \operatorname{ev}(x,p)=x(p),
 \label{eq:brane-evaluation}
\end{equation}
be the evaluation map.  We normalize the Haar measure $d\theta$ on $\mathbb T$
to have total mass 1.  The normalized spherical measure in the latitude
variable is
\begin{equation}
 \rho(t)\,dt=\frac{1}{2}\cos t\,dt,
 \label{eq:latitude-measure}
\end{equation}
so that
\[
 \int_{-\pi/2}^{\pi/2}\rho(t)\,dt=1.
\]

For $p\in S^2$, let $\operatorname{ev}_p:S^2M\to M$ be evaluation at $p$.
Given $\omega\in\Omega^j(M)$, we define its {\bf average on brane space $S^2M$} by
\begin{equation}
 \overline\omega=
 \int_{-\pi/2}^{\pi/2}\int_{\mathbb T}
 \operatorname{ev}_{c_t(\theta)}^*\omega\,
 d\theta\,\rho(t)\,dt
 \in\Omega^j(S^2M).
 \label{eq:average-form}
\end{equation}
This operation preserves degree.  It should not be confused with {\bf brane
transgression}, which is obtained by integration along the fibre of
$S^2\times S^2M\to S^2M$ and lowers degree by two.

There is also a transgression associated with a single latitude.  For fixed $t$,
consider
\begin{equation}
 \operatorname{ev}_t:S^2M\times\mathbb T\longrightarrow M,
 \qquad
 \operatorname{ev}_t(x,\theta)=x(c_t(\theta)).
\end{equation}
The {\bf latitude transgression} of $\omega$ is
\begin{equation}
 \tau_t(\omega)=\int_{\mathbb T}\operatorname{ev}_t^*\omega
 \in\Omega^{j-1}(S^2M).
 \label{eq:latitude-transgression}
\end{equation}
The orientation of the fibre $\mathbb T$ is the one determined by $V$.

The elementary identities below are the reason for introducing both operations.
They will reappear in the construction of the connection on the average holonomy
line bundle.

\begin{lemma}
For every $\omega\in\Omega^\bullet(M)$, one has
\begin{equation}
 d\overline\omega=\overline{d\omega},
 \qquad
 \mathcal L_K\overline\omega=0,
 \qquad
 \iota_K\overline\omega=
 \int_{-\pi/2}^{\pi/2}
 \tau_t(\omega)\rho(t)\,dt.
 \label{eq:average-identities}
\end{equation}
\end{lemma}

\begin{proof}
The first identity follows by differentiating under the integral sign.  The
second is a consequence of the invariance of Haar measure.  For the last
identity, contraction with $K$ inserts the infinitesimal rotation vector $V$ in
the evaluation map.  Integration in the orbit direction is therefore exactly
the latitude transgression in \eqref{eq:latitude-transgression}.  Averaging over
$t$ gives the stated formula.
\end{proof}

The preceding lemma passes from differential forms on $M$ to invariant forms on
$S^2M$.  We now carry out the corresponding construction for gerbe data.

\subsection{The average holonomy line bundle arising from gerbe data}

For an open set $\sU\subset M$, write
\begin{equation}
 S^2\sU=C^\infty(S^2,\sU)\subset S^2M.
\end{equation}

\begin{definition}
An open cover $\{\sU_\alpha\}$ of $M$ is called a
{\bf brane Brylinski cover} if the sets
$\{S^2\sU_\alpha\}$ cover $S^2M$.
\end{definition}

We choose this cover relative to the gerbe that enters the construction.
Let $\mathcal G$ be a gerbe with connection on $M$, with curvature $H$
and integral Dixmier-Douady class $h$.  We require only
\[
 h|_{\sU_\alpha}=0
 \quad\text{in }H^3(\sU_\alpha;\mathbb Z).
\]
Such image-containing neighborhoods exist for every smooth sphere map.
Indeed, $K=x(S^2)$ is compact and has Hausdorff dimension at most two,
and hence covering dimension at most two.  Consequently
$\check H^3(K;\mathbb Z)=0$.  Continuity of \v{C}ech cohomology gives
\[
 \check H^3(K;\mathbb Z)
 =\varinjlim_{K\subset\sU}H^3(\sU;\mathbb Z),
\]
where $\sU$ runs over open neighborhoods in $M$; see Chapter 6 in \cite{Spanier}.
The restriction of the fixed class $h$ therefore vanishes on some such
neighborhood.  Choosing one for each $x$ gives the required brane cover.
This argument uses smoothness of $x$ and does not require its image
to be an embedded surface.

Choose a trivialization of $\mathcal G$ on each $\sU_\alpha$, with
curving $B_\alpha$.  Here triviality is topological: as a gerbe with
connection the restriction is isomorphic to the trivial gerbe
$\mathcal I_{B_\alpha}$, and $dB_\alpha=H|_{\sU_\alpha}$.
The differences between these trivializations are line bundles with
connection $(L_{\alpha\beta},\nabla^{L_{\alpha\beta}})$ on
$\sU_{\alpha\beta}=\sU_\alpha\cap\sU_\beta$.  We retain their
connection-preserving products
\[
 \mu_{\alpha\beta\gamma}:
 L_{\alpha\beta}\otimes L_{\beta\gamma}
 \longrightarrow L_{\alpha\gamma},
\]
which are associative on quadruple intersections.  This is the usual
descent description of a gerbe \cite{Bry,MS}.  No vanishing of the
higher cohomology of finite intersections is needed; the overlap line
bundles need not be trivial.

In this notation the connection data are
\[
 \bigl(H,B_\alpha,F_{\alpha\beta},
 (L_{\alpha\beta},\nabla^{L_{\alpha\beta}})\bigr),
\]
and their descent relations are
\begin{equation}
\label{eq:gerbe-triple-brane}
\begin{split}
 &H=dB_\alpha
 \qquad\text{on }\sU_\alpha,\\
 &B_\beta-B_\alpha
 =F_{\alpha\beta}
 =\left(\nabla^{L_{\alpha\beta}}\right)^2
 \qquad\text{on }\sU_{\alpha\beta},\\
 &(L_{\alpha\beta},\nabla^{L_{\alpha\beta}})
 \otimes(L_{\beta\gamma},\nabla^{L_{\beta\gamma}})
 \otimes(L_{\gamma\alpha},\nabla^{L_{\gamma\alpha}})
 \simeq(\CC,d)
 \qquad\text{on }\sU_{\alpha\beta\gamma}.
\end{split}
\end{equation}
The last identification is the one supplied by the gerbe product.

Let
\[
 x\in S^2\sU_\alpha\cap S^2\sU_\beta.
\]
For each latitude, define
\begin{equation}
 h_{\alpha\beta}(t;x)=
 \operatorname{hol}_{\nabla^{L_{\alpha\beta}}}
 \bigl(\theta\longmapsto x(c_t(\theta))\bigr).
 \label{eq:latitude-holonomy}
\end{equation}
Thus $h_{\alpha\beta}(t;x)$ is the holonomy of
$(L_{\alpha\beta},\nabla^{L_{\alpha\beta}})$ around the image under $x$ of the
latitude at $t$.

At the south pole the latitude is a constant loop, and hence
\[
 h_{\alpha\beta}(-\pi/2;x)=1.
\]
The path lifting property of the exponential map gives a unique continuous
logarithm
\begin{equation}
\begin{split}
 \ell_{\alpha\beta}(t;x)&\in\sqrt{-1}\mathbb R,
 \qquad
 \ell_{\alpha\beta}(-\pi/2;x)=0,\\
 \exp\bigl(\ell_{\alpha\beta}(t;x)\bigr)
 &=h_{\alpha\beta}(t;x).
\end{split}
 \label{eq:normalized-log-holonomy}
\end{equation}
The normalization at the south pole removes the usual ambiguity of the
logarithm.  The resulting function depends smoothly on both $x$ and $t$.

We now average this logarithmic holonomy over the latitudes and return to the
circle by exponentiation.  Define
\begin{equation}
 g_{\alpha\beta}(x)=
 \exp\left(
 \int_{-\pi/2}^{\pi/2}
 \ell_{\alpha\beta}(t;x)\rho(t)\,dt
 \right)
 \label{eq:average-transition}
\end{equation}
on $S^2\sU_\alpha\cap S^2\sU_\beta$.

\begin{proposition}
The functions $g_{\alpha\beta}$ are $\mathbb T$ invariant and satisfy
\begin{equation}
 g_{\alpha\beta}g_{\beta\gamma}g_{\gamma\alpha}=1.
 \label{eq:average-cocycle}
\end{equation}
They therefore define the transition functions of a complex line bundle
\begin{equation}
 \mathcal L_H\longrightarrow S^2M.
 \label{eq:average-line-bundle}
\end{equation}
We call $\mathcal L_H$ the
{\bf average latitude holonomy line bundle on the brane space}.
\end{proposition}

\begin{proof}
Rotating a parametrized circle does not change its holonomy.  The functions
$h_{\alpha\beta}$, and hence $g_{\alpha\beta}$, are therefore invariant under
$\mathbb T$.

The gerbe cocycle relation in \eqref{eq:gerbe-triple-brane} gives
\begin{equation}
 h_{\alpha\beta}(t;x)h_{\beta\gamma}(t;x)
 h_{\gamma\alpha}(t;x)=1.
\end{equation}
It follows that
\begin{equation}
 \ell_{\alpha\beta}(t;x)
 +\ell_{\beta\gamma}(t;x)
 +\ell_{\gamma\alpha}(t;x)
 \in2\pi\sqrt{-1}\mathbb Z.
\end{equation}
The left hand side is continuous in $t$ and vanishes at the south pole, so it vanishes for every $t$.  Averaging and exponentiating proves
\eqref{eq:average-cocycle}.
\end{proof}

Let $s_\alpha$ be the local frame of $\mathcal L_H$ over $S^2\sU_\alpha$
determined by the convention
\begin{equation}
 s_\alpha=g_{\alpha\beta}s_\beta.
 \label{eq:frame-convention}
\end{equation}
The frames may be chosen to be $\mathbb T$-invariant.

The transition functions alone do not yet retain the differential information
contained in the gerbe connection. The next proposition supplies the missing
connection and identifies its curvature.

\begin{proposition}
\label{connection}
The one-forms
\begin{equation}
 A_\alpha=-\iota_K\overline{B_\alpha}
 \label{eq:average-connection-form}
\end{equation}
are the local connection forms of a $\mathbb T$-invariant connection
$\nabla^{\mathcal L_H}$ on $\mathcal L_H$.  Its curvature is
\begin{equation}
 (\nabla^{\mathcal L_H})^2=\iota_K\overline H.
 \label{eq:average-curvature}
\end{equation}
\end{proposition}

\begin{proof}
The variation formula for line-bundle holonomy \cite{Bry} gives
\begin{equation}
 d\ell_{\alpha\beta}(t)=\tau_t(F_{\alpha\beta}).
 \label{eq:brylinski-dlog}
\end{equation}
After averaging in $t$, the last identity in
\eqref{eq:average-identities} yields
\begin{equation}
 d\log g_{\alpha\beta}
 =\iota_K\overline{F_{\alpha\beta}}.
 \label{eq:dlog-average-transition}
\end{equation}

Since $B_\beta-B_\alpha=F_{\alpha\beta}$, we obtain
\begin{equation}
 A_\beta-A_\alpha=-d\log g_{\alpha\beta}.
\end{equation}
This is exactly the transformation law for connection forms in the frames fixed
by \eqref{eq:frame-convention}.  It follows that the $A_\alpha$ define a global
$\mathbb T$ invariant connection on $\mathcal L_H$.

Finally,
\[
 \mathcal L_K\overline{B_\alpha}=0.
\]
Cartan's formula and $dB_\alpha=H$ therefore give
\begin{equation}
 dA_\alpha
 =-d\iota_K\overline{B_\alpha}
 =\iota_Kd\overline{B_\alpha}
 =\iota_K\overline H.
\end{equation}
This proves the curvature formula.
\end{proof}

The construction is compatible with a change of the chosen gerbe
trivializations.  Their difference on $\sU_\alpha$ is a line bundle
with connection.  Its normalized average latitude holonomy supplies
the change of frame of $\mathcal L_H$, and the same variation formula
identifies the connections.  Passing to a common image-containing
refinement therefore gives an isomorphic line bundle with connection.

Although its connection is generally curved, the underlying line bundle has a
particularly simple topology.

\begin{proposition}
The complex line bundle $\mathcal L_H$ is topologically trivial.
\end{proposition}

\begin{proof}
For $r\in[0,1]$, put
\begin{equation}
 g^{(r)}_{\alpha\beta}(x)=
 \exp\left(
 r\int_{-\pi/2}^{\pi/2}
 \ell_{\alpha\beta}(t;x)\rho(t)\,dt
 \right).
\end{equation}
The proof of \eqref{eq:average-cocycle} applies without change and shows that
$\{g^{(r)}_{\alpha\beta}\}$ is a system of transition functions for every $r$.
At $r=0$ these are the transition functions of the trivial line bundle, while at
$r=1$ they are those of $\mathcal L_H$.  This homotopy proves the assertion.
\end{proof}

\begin{remark}
The point of $\mathcal L_H$ is not its first Chern class, but its distinguished
connection.  In general this connection is not flat, and its curvature is the
form $\iota_K\overline H$ in \eqref{eq:average-curvature}.  {\em The topological
triviality reflects the elementary fact that every orbit circle on $S^2$ bounds
a disc.  This is one of the features which distinguishes the present
construction from the double loop construction of \cite{HM24}}.
\end{remark}

We have now obtained the two terms which must cancel in the equivariant theory:
the curvature of $\nabla^{\mathcal L_H}$ and the contraction of the averaged
3-form.  The natural language for this cancellation is that of superconnections.

\subsection{The equivariantly flat superconnection and the exotic twisted theory on brane space}
\label{subsec:brane-periodic-completion}

Let $u$ have degree $2$.  We spell out the completion carefully, since the localization argument depends on it.  For a space $Y$ and a
coefficient bundle $\xi$, the component of total degree $k$ is
\[
 \bigl(\Omega^\bullet(Y,\xi)[[u,u^{-1}]]\bigr)^k
 =
 \prod_{\substack{j\geq0\\j\equiv k\;(\mathrm{mod}\;2)}}
 u^{(k-j)/2}\Omega^j(Y,\xi).
\]
The resulting graded space is the direct sum of these components over
$k\in\mathbb Z$.  We also write
$\Omega^\bullet(Y,\xi)\widehat\otimes R$ for it, where
$R=\mathbb C[u,u^{-1}]$ is the graded Laurent polynomial ring.
Thus the completion is taken in ordinary form degree, not in a ring
of unrestricted two-sided scalar series.  Products and the operators
used below are defined componentwise.  In each ordinary form degree
the products involve only finitely many terms.  The same convention
applies to invariant forms and to forms with trivial coefficients.
For finite dimensional $Y$ the product at fixed total degree is
finite, so the completed tensor product agrees with the ordinary one.

On the completed periodic space
\begin{equation}
 \Omega^\bullet(S^2M,\mathcal L_H)[[u,u^{-1}]],
\end{equation}
consider the odd operator
\begin{equation}
 Q_H
 =\nabla^{\mathcal L_H}
 -u\iota_K
 +u^{-1}\overline H.
 \label{eq:brane-superconnection}
\end{equation}
The curvature identity above is exactly what is required for equivariant
flatness.

\begin{theorem}
One has
\begin{equation}
 Q_H^2+u\mathcal L_K^{\mathcal L_H}=0.
 \label{eq:brane-equivariant-flatness}
\end{equation}
In particular, $Q_H^2=0$ on the invariant subspace
\begin{equation}
 \Omega^\bullet(S^2M,\mathcal L_H)^{\mathbb T}
 [[u,u^{-1}]].
\end{equation}
\end{theorem}

\begin{proof}
Work over $S^2\sU_\alpha$ and use the invariant frame $s_\alpha$.  In this
frame, the operator is
\begin{equation}
 Q_H
 =d-\iota_K\overline{B_\alpha}
 -u\iota_K
 +u^{-1}d\overline{B_\alpha}.
\end{equation}
Equivalently,
\begin{equation}
 Q_H=
 e^{-u^{-1}\overline{B_\alpha}}
 (d-u\iota_K)
 e^{u^{-1}\overline{B_\alpha}}.
 \label{eq:local-conjugation}
\end{equation}
Since $\overline{B_\alpha}$ is $\mathbb T$ invariant, Cartan's formula gives
\begin{equation}
 Q_H^2=
 e^{-u^{-1}\overline{B_\alpha}}
 (d-u\iota_K)^2
 e^{u^{-1}\overline{B_\alpha}}
 =-u\mathcal L_K.
\end{equation}
The local frames are invariant, so the operator on the right is
$-u\mathcal L_K^{\mathcal L_H}$.  These local identities agree on intersections
and hence give the asserted global formula.
\end{proof}

\begin{definition}
The {\bf completed periodic exotic twisted equivariant cohomology of the brane
space} is the cohomology of the complex
\begin{equation}
 \left(
 \Omega^\bullet(S^2M,\mathcal L_H)^{\mathbb T}
 [[u,u^{-1}]],
 Q_H
 \right).
 \label{eq:brane-exotic-complex}
\end{equation}
We denote it by
\begin{equation}
 h^\bullet_{\mathbb T}
 \bigl(S^2M,\nabla^{\mathcal L_H}:\overline H\bigr).
 \label{eq:brane-exotic-cohomology}
\end{equation}
\end{definition}

Thus the gerbe on $M$ gives rise, by averaging its latitude holonomy, to a
natural twisted equivariant theory on $S^2M$.  The construction retains the
local geometry of the gerbe through $\nabla^{\mathcal L_H}$, while
\eqref{eq:brane-equivariant-flatness} makes the resulting complex accessible to
localization.  This is the {\em  passage from target space geometry to brane space
geometry} which underlies the next section.

\section{Localization and path space contraction}
\label{sec:brane-localization}

The exotic twisted theory constructed in the preceding section lives naturally
on the infinite-dimensional brane space $S^2M$. The purpose of the present
section is to examine the cohomological content of this construction, and to
show that, once we pass to the completed periodic localization, the theory
is isomorphic to the ordinary $H$-twisted theory of $M$
itself. Seen in this light, the whole construction may be regarded quite
naturally as a delocalization of the twisted cohomology theory of the
constant brane space, that is to say, of the manifold $M$ itself, out onto
the full brane space $S^2M$.

The argument proceeds in two stages. We begin by recalling the localization
principle for a general circle manifold equipped with an equivariantly flat
superconnection, developed in \cite{HM15,HM24}. Armed with this principle, we
then identify the fixed point set of $S^2M$ with the free path space $PM$,
and observe that this free path space contracts naturally onto the space of
constant paths, i.e. $M$ itself.

\subsection{The localization principle}

We begin by recalling the exotic twisted completed periodic localization
theory developed in \cite{HM15,HM24}. Rather than reproduce the theory in
full, we shall content ourselves with drawing out its essential features,
retaining only what is needed for our later purposes; the reader who wishes
to see the complete picture, with all its proofs and refinements, is
referred to \cite{HM15,HM24}.

Let $Y$ be a possibly infinite dimensional
$\mathbb T$-manifold, and let $K$ be the vector field which generates the circle
action.  We work with smooth differential forms and assume the
availability of smooth partitions of unity and an invariant weak
Riemannian metric, as for the smooth mapping spaces considered here.
All completions have the meaning fixed in
Section~\ref{subsec:brane-periodic-completion}.  Suppose that $\xi\to Y$ is a $\mathbb T$ equivariant complex line
bundle with invariant connection $\nabla^\xi$, and that
$\chi\in\Omega^3(Y)$ is closed and invariant.  Introduce the odd operator
\begin{equation}
 D_{\xi,\chi}
 =\nabla^\xi-u\iota_K+u^{-1}\chi
 \label{eq:general-superconnection}
\end{equation}
on $\Omega^\bullet(Y,\xi)[[u,u^{-1}]]$, where $u$ has degree $2$.  We assume
that
\begin{equation}
 D_{\xi,\chi}^2+u\mathcal L_K^\xi=0.
 \label{eq:general-flatness}
\end{equation}
Thus $D_{\xi,\chi}$ is an equivariantly flat superconnection and
defines a differential on the invariant subspace.  Its cohomology is denoted by
\begin{equation}
 \begin{split}
 h^\bullet_{\mathbb T}(Y,\nabla^\xi:\chi)
 &=
 H^\bullet\Bigl(
 \Omega^\bullet(Y,\xi)^{\mathbb T}[[u,u^{-1}]],
 D_{\xi,\chi}
 \Bigr).
 \end{split}
 \label{eq:general-exotic-theory}
\end{equation}

The fixed point theorem requires a mild regularity condition near the fixed
point set.  It allows us to compare a tubular neighborhood with the fixed points themselves without losing control of the coefficient data.

\begin{definition}
Let $F=Y^{\mathbb T}$.  The $\mathbb T$-manifold $Y$ is called
\emph{strongly regular} if $F$ is a smooth submanifold and has an invariant
neighborhood $N$ with the following properties.
\begin{enumerate}[(i)]
\item The neighborhood $N$ admits an invariant good cover.
\item There are an invariant projection $p:N\to F$ and an equivariant homotopy
\begin{equation}
 g:N\times[0,1]\longrightarrow N
\end{equation}
such that
\begin{equation}
 g_0=i\circ p,
 \qquad
 g_1=\operatorname{id}_N,
 \qquad
 p(g_t(y))=p(y).
 \label{eq:strong-regular-homotopy}
\end{equation}
Here $i:F\to N$ is the inclusion.
\end{enumerate}
\end{definition}

The localization argument has four ingredients.  The first explains
why the completed periodic theory is concentrated at the fixed points.

\begin{lemma}[Acyclicity away from the fixed points]
\label{lem:no-fixed-points}
If the $\mathbb T$-action on $Y$ has no fixed points, then
\begin{equation}
 h^\bullet_{\mathbb T}(Y,\nabla^\xi:\chi)=0.
\end{equation}
\end{lemma}

\begin{proof}
Choose an invariant metric on $Y$.  Since $K$ is nowhere zero, the one-form
\begin{equation}
 \vartheta=\frac{K^\flat}{\lVert K\rVert^2}
\end{equation}
is invariant and satisfies $\iota_K\vartheta=1$.  Put
$\delta=d-u\iota_K$.  In the completed periodic algebra the even form
\begin{equation}
 \Gamma=(\delta\vartheta)^{-1}
 =-u^{-1}\sum_{j=0}^{\infty}u^{-j}(d\vartheta)^j
 \label{eq:inverse-delta-theta}
\end{equation}
is well defined.  Since $\mathcal L_K\vartheta=0$, one has
$\delta\Gamma=0$.  The odd form
\[
 \omega=\vartheta\Gamma
\]
therefore satisfies $\delta\omega=1$.  If
$a\in\Omega^\bullet(Y,\xi)^{\mathbb T}[[u,u^{-1}]]$, graded commutativity gives
\begin{equation}
 D_{\xi,\chi}(\omega a)+\omega D_{\xi,\chi}a=a.
\end{equation}
Multiplication by $\omega$ is consequently a contracting homotopy.
\end{proof}

The geometric series in \eqref{eq:inverse-delta-theta} is where the completion enters the argument.  It provides an inverse to the equivariant
Euler term wherever the generating vector field does not vanish.  This is the local algebra behind Borel-Witten localization.

The second ingredient is the familiar gluing property.

\begin{lemma}[Mayer-Vietoris]
\label{lem:mayer-vietoris}
Suppose that $Y=\sU\cup\sV$, where $\sU$, $\sV$, and
$\sU\cap\sV$ are invariant open subsets.  The complex
\eqref{eq:general-exotic-theory} gives rise to a Mayer-Vietoris long exact
sequence.
\end{lemma}

\begin{proof}
Choose a partition of unity subordinate to $\{\sU,\sV\}$ and average it over
$\mathbb T$.  The usual restriction sequence of invariant differential forms
with values in $\xi$ is then short exact, and all its maps commute with
$D_{\xi,\chi}$.  The associated long exact sequence is the required one.
\end{proof}

The remaining two ingredients concern homotopy.  We first record the easy case
in which all the coefficient data are pulled back from one end of a cylinder.

\begin{lemma}[Cylinder invariance]
\label{lem:completed-cylinder-invariance}
Let $\pi:Y\times[0,1]\to Y$ be the projection, with $\mathbb T$ acting trivially
on the second factor.  Pullback induces an isomorphism
\begin{equation}
 \pi^*:
 h^\bullet_{\mathbb T}(Y,\nabla^\xi:\chi)
 \xrightarrow{\ \cong\ }
 h^\bullet_{\mathbb T}
 (Y\times[0,1],\pi^*\nabla^\xi:\pi^*\chi).
 \label{eq:cylinder-invariance}
\end{equation}
\end{lemma}

\begin{proof}
The usual integration operator in the interval direction gives a chain
homotopy between the identity and $\pi^*i_0^*$.  Since
$\iota_K$ and multiplication by $\pi^*\chi$ are horizontal, the same operator
is a chain homotopy for the completed periodic differential.  Thus $\pi^*$ and
$i_0^*$ induce inverse maps on cohomology.
\end{proof}

For localization, however, the pulled back bundle $g^*\xi$ on a homotopy
cylinder need not be the pullback of its restriction at either endpoint as a
bundle with connection.  Moreover, the 3-form $g^*\chi$ also varies with
the homotopy parameter.  The following endpoint lemma is the delicate part of
the theory.  It is the result which allows a geometric deformation to carry
the complete twisting data with it.

\begin{lemma}[Endpoint invariance]
\label{lem:equivariant-endpoint-invariance}
Let $Y$ admit an invariant good cover, and let
$i_t:Y\to Y\times[0,1]$ denote inclusion at $t$.  Suppose that
$\Xi\to Y\times[0,1]$ is a $\mathbb T$-equivariant complex line bundle with
invariant connection $\nabla^\Xi$, and that
$\widetilde\chi\in\Omega^3(Y\times[0,1])$ is closed and invariant.  If
\begin{equation}
 \left(\nabla^\Xi-u\iota_K+u^{-1}\widetilde\chi\right)^2
 +u\mathcal L_K^\Xi=0,
 \label{eq:cylinder-general-flatness}
\end{equation}
then there is an isomorphism
\begin{equation}
 \rho_{10}:
 h^\bullet_{\mathbb T}
 (Y,i_0^*\nabla^\Xi:i_0^*\widetilde\chi)
 \xrightarrow{\ \cong\ }
 h^\bullet_{\mathbb T}
 (Y,i_1^*\nabla^\Xi:i_1^*\widetilde\chi).
 \label{eq:endpoint-invariance}
\end{equation}
\end{lemma}

\begin{proof}
We recall the construction, since it contains the essential cancellation.
Parallel transport for $\nabla^\Xi$ in the interval direction identifies the
underlying coefficient bundles at the two ends.  Let $\tau(\alpha)$ denote the
parallel extension to $Y\times[0,1]$ of a form
$\alpha\in\Omega^\bullet(Y,i_0^*\Xi)$.

Let $K_0$ be the usual homotopy operator in the interval direction.  If a
differential form is written as
\[
 \alpha(t)+dt\wedge\beta(t),
\]
then
\begin{equation}
 K_0\bigl(\alpha(t)+dt\wedge\beta(t)\bigr)
 =\int_0^t\beta(s)\,ds.
 \label{eq:interval-homotopy-operator}
\end{equation}
It satisfies
\begin{equation}
 dK_0+K_0d
 =\operatorname{id}-\pi^*i_0^*.
 \label{eq:interval-homotopy-formula}
\end{equation}
The endpoint map is represented by
\begin{equation}
 \rho_{10}(\alpha)
 =
 i_1^*\left(
 e^{-u^{-1}K_0\widetilde\chi}\tau(\alpha)
 \right).
 \label{eq:endpoint-map-formula}
\end{equation}

In the parallel-transport identification, write the connection as
$\pi^*\nabla_0+A_t+dt\,\partial_t$, where $A_0=0$, and write
$\widetilde\chi=\chi_t+dt\wedge\psi_t$.  Put
$B_t=\int_0^t\psi_s\,ds$.  Closedness of $\widetilde\chi$ and
the curvature part of \eqref{eq:cylinder-general-flatness} give
\[
 d_YB_t=\chi_t-\chi_0,\qquad
 \partial_tA_t=-\iota_K\psi_t,\qquad A_t=-\iota_KB_t.
\]
The sign in the middle identity comes from
$\iota_K(dt\wedge\psi_t)=-dt\wedge\iota_K\psi_t$.
Consequently the endpoint operators satisfy
\[
 D_t=e^{-u^{-1}B_t}D_0e^{u^{-1}B_t}.
\]
This identity accounts for both variations: $d_YB_t$ corrects the
flux, and $-\iota_KB_t$ corrects the connection.  Since the
identification is by equivariant parallel transport, it is global and
preserves invariant forms.  At $t=1$ it proves
\eqref{eq:endpoint-map-formula}; multiplication by the opposite
exponential followed by inverse parallel transport gives its inverse.
\end{proof}

\begin{remark}
Lemma \ref{lem:equivariant-endpoint-invariance} is stronger than ordinary
homotopy invariance.  The bundle, its connection, and the twisting three-form
are all allowed to vary.  Parallel transport controls the bundle, the
exponential $e^{-u^{-1}K_0\widetilde\chi}$ controls the 3-form, and
equivariant flatness is precisely the identity which makes the two corrections
compatible.  It is this mechanism that permits localization with nontrivial coefficient line bundles.
\end{remark}

We can now compare a strongly regular neighborhood with its fixed point set.

\begin{lemma}[Equivariant retraction]
\label{lem:equivariant-retraction}
Let $N$ be an invariant neighborhood of $F$ with the data in
\eqref{eq:strong-regular-homotopy}.  Restriction induces an isomorphism
\begin{equation}
 \begin{split}
 i^*:h^\bullet_{\mathbb T}(N,\nabla^\xi:\chi)
 \longrightarrow
 H^\bullet\Bigl(
 \Omega^\bullet(F,i^*\xi)[[u,u^{-1}]],
 i^*\nabla^\xi+u^{-1}i^*\chi
 \Bigr).
 \end{split}
 \label{eq:retraction-isomorphism}
\end{equation}
\end{lemma}

\begin{proof}
The relation $p\circ g_t=p$ shows, by cylinder invariance, that $p^*$ and $i^*$
are inverse on the theory formed from the pulled back data
$p^*i^*(\xi,\nabla^\xi,\chi)$.  Pulling the original data back by
$g:N\times[0,1]\to N$ gives an equivariantly flat family whose restriction at
$t=0$ is the pulled back fixed point data and whose restriction at $t=1$ is the
original data on $N$.  Lemma \ref{lem:equivariant-endpoint-invariance}
identifies the two endpoint theories.  The homotopy is stationary
on $F$, so this identification restricts to the identity there.
Composing with the pulled-back contraction therefore gives restriction
$i^*$ itself as the isomorphism in
\eqref{eq:retraction-isomorphism}.
\end{proof}

The localization theorem now follows from these structural facts.

\begin{theorem}[Exotic twisted Borel-Witten localization]
\label{thm:general-localization}
If $Y$ is strongly regular, restriction to the fixed point set induces an
isomorphism
\begin{equation}
 \begin{split}
 i^*:h^\bullet_{\mathbb T}(Y,\nabla^\xi:\chi)
 \longrightarrow
 H^\bullet\Bigl(
 \Omega^\bullet(F,i^*\xi)[[u,u^{-1}]],
 i^*\nabla^\xi+u^{-1}i^*\chi
 \Bigr).
 \end{split}
 \label{eq:general-localization-isomorphism}
\end{equation}
\end{theorem}

\begin{proof}
Let $N$ be the invariant neighborhood appearing in the definition of strong
regularity.  Apply the Mayer-Vietoris sequence to
\begin{equation}
 Y=N\cup(Y\setminus F).
\end{equation}
Both $Y\setminus F$ and $N\setminus F$ are fixed point free.  Their completed
periodic cohomology therefore vanishes by Lemma
\ref{lem:no-fixed-points}.  Mayer-Vietoris identifies the theory of $Y$ with
that of $N$, and Lemma \ref{lem:equivariant-retraction} identifies the latter
with the theory of $F$.
\end{proof}

The theorem holds for general $Y$; on the brane space it takes a concrete geometric form.  The fixed maps are not yet the points of $M$.  They are maps
which may vary in the latitude direction, and hence form a path space.  We now
make this identification precise.

\subsection{Path space as the fixed point set of the brane space}

Let
\begin{equation}
 PM=C^\infty([-1,1],M)
 \label{eq:free-path-space}
\end{equation}
be the free path space.  If $z:S^2\to[-1,1]$ is the height function, define
\begin{equation}
 j:PM\longrightarrow S^2M,
 \qquad
 j(\gamma)=\gamma\circ z.
 \label{eq:path-fixed-inclusion}
\end{equation}

\begin{proposition}
\label{prop:brane-fixed-points}
The fixed point set of the circle action on $S^2M$ is $j(PM)$.  Thus
\begin{equation}
 (S^2M)^{\mathbb T}\simeq PM.
 \label{eq:fixed-path-space}
\end{equation}
Moreover, $S^2M$ is strongly regular in the sense of Theorem
\ref{thm:general-localization}.
\end{proposition}

\begin{proof}
A map $x:S^2\to M$ is fixed if and only if it is constant on every orbit circle.
It therefore factors uniquely through the height function.  Smoothness at the
two poles is equivalent to smoothness of the resulting map as a function of the
height variable.  This proves \eqref{eq:fixed-path-space}.

Choose a Riemannian metric on $M$.  Along a fixed map $x$,
averaging sections over the source circle is a continuous projection
onto the invariant sections.  Its kernel is an invariant complement
to $T_xj(PM)$ and thus gives the normal bundle.

For completeness, the pointwise exponential construction can be made
uniform near each fixed map as follows.  For a nearby sphere map $y$,
each orbit image lies in a sufficiently small convex normal ball.
Take its Riemannian center of mass with normalized circle measure
\cite{Karcher}.  The center depends smoothly on the orbit image and
is characterized by the vanishing average of its logarithmic
displacements.  As the orbit varies, these centers give a smooth
invariant sphere map $p(y)$, hence a point of $j(PM)$.  Smooth
dependence follows from the finite dimensional implicit function
theorem in the convex balls, applied pointwise to smooth families.

The section $\zeta_y(q)=\exp_{p(y)(q)}^{-1}(y(q))$ has zero orbit
average.  Conversely, a sufficiently small section with zero orbit
average has $p(y)$ as the center of its pointwise exponential.
After shrinking the neighborhood, these constructions are inverse
normal coordinates.  Choose their domain to be star-shaped in each
normal fibre.  Scaling $\zeta_y$ by $s\in[0,1]$ then defines an
equivariant homotopy $g_s$ with
$g_0=ip$, $g_1=\operatorname{id}$ and $p g_s=p$.
The allowable radius may depend on the fixed map; a uniform radius
on the whole path space is unnecessary.

Finally, $PM$ has a good cover by sufficiently small pointwise
geodesically convex path neighborhoods.  Nonempty finite
intersections contract by pointwise geodesics to any chosen path
in the intersection.  Pulling this cover back by $p$ and using the
normal contraction gives an invariant good cover of the neighborhood.
This verifies the strong regularity required above.
\end{proof}

The restriction of the coefficient data is especially simple.  On the fixed
point set every orbit circle is constant.  Its holonomy is therefore trivial,
and the normalized logarithm used in the construction of $\mathcal L_H$
vanishes.  Consequently,
\begin{equation}
 j^*g_{\alpha\beta}=1,
 \qquad
 j^*\mathcal L_H\simeq PM\times\mathbb C,
 \qquad
 j^*\nabla^{\mathcal L_H}=d.
 \label{eq:fixed-line-trivialization}
\end{equation}
Under the coordinate $z=\sin t$, the restriction of $\overline H$ is
\begin{equation}
 H_P=\frac{1}{2}\int_{-1}^{1}\operatorname{ev}_z^*H\,dz
 \in\Omega^3(PM),
 \label{eq:path-average-H}
\end{equation}
where $\operatorname{ev}_z(\gamma)=\gamma(z)$.

\begin{theorem}
\label{thm:brane-to-path}
Restriction to the fixed point set induces an isomorphism
\begin{equation}
 \begin{split}
 j^*:h^\bullet_{\mathbb T}
 (S^2M,\nabla^{\mathcal L_H}:\overline H)
 \longrightarrow
 H^\bullet\Bigl(
 \Omega^\bullet(PM)[[u,u^{-1}]],
 d+u^{-1}H_P
 \Bigr).
 \end{split}
 \label{eq:brane-to-path-isomorphism}
\end{equation}
\end{theorem}

\begin{proof}
Apply Theorem \ref{thm:general-localization} with
\[
 Y=S^2M,
 \qquad
 \xi=\mathcal L_H,
 \qquad
 \chi=\overline H.
\]
The required equivariant flatness is
\eqref{eq:brane-equivariant-flatness}.  Proposition
\ref{prop:brane-fixed-points}, together with
\eqref{eq:fixed-line-trivialization} and \eqref{eq:path-average-H}, identifies
the fixed point complex with the complex on the right hand side of
\eqref{eq:brane-to-path-isomorphism}.
\end{proof}

Localization has reduced the brane space to the path space.  A final elementary
homotopy now reduces the path space to $M$.  The only subtlety is that the averaged 3-form $H_P$ is not literally the pullback of $H$ by
evaluation at the midpoint.  The two forms differ by an exact term, and this
term produces the gauge transformation below.

\subsection{Contraction of the path space}

Let
\begin{equation}
 e_0:PM\longrightarrow M,
 \qquad
 e_0(\gamma)=\gamma(0),
\end{equation}
and let $c:M\to PM$ be the inclusion of the constant paths.  Define
\begin{equation}
 r_s:PM\longrightarrow PM,
 \qquad
 (r_s\gamma)(z)=\gamma(sz),
 \qquad
 0\leq s\leq1.
 \label{eq:path-retraction}
\end{equation}
Then
\[
 r_1=\operatorname{id}_{PM},
 \qquad
 r_0=c\circ e_0.
\]

Let $r:PM\times[0,1]\to PM$ denote this homotopy, and put
\begin{equation}
 \eta=\int_0^1\iota_{\partial_s}r^*H_P\,ds
 \in\Omega^2(PM).
 \label{eq:path-eta}
\end{equation}

\begin{proposition}
One has
\begin{equation}
 H_P-e_0^*H=d\eta.
 \label{eq:path-twist-homotopy}
\end{equation}
Multiplication by $e^{-u^{-1}\eta}$ defines an isomorphism of complexes
\begin{equation}
 \begin{split}
 \Psi:\quad&
 \Bigl(
 \Omega^\bullet(PM)[[u,u^{-1}]],
 d+u^{-1}e_0^*H
 \Bigr)\\
 &\longrightarrow
 \Bigl(
 \Omega^\bullet(PM)[[u,u^{-1}]],
 d+u^{-1}H_P
 \Bigr).
 \end{split}
 \label{eq:path-gauge-isomorphism}
\end{equation}
\end{proposition}

\begin{proof}
The homotopy formula gives
\begin{equation}
 r_1^*H_P-r_0^*H_P=d\eta.
\end{equation}
The measure in \eqref{eq:path-average-H} has total mass one, and hence
$c^*H_P=H$.  It follows that $r_0^*H_P=e_0^*H$, which proves
\eqref{eq:path-twist-homotopy}.  A direct calculation gives
\begin{equation}
 (d+u^{-1}H_P)e^{-u^{-1}\eta}
 =
 e^{-u^{-1}\eta}(d+u^{-1}e_0^*H).
\end{equation}
This proves \eqref{eq:path-gauge-isomorphism}.
\end{proof}

The maps $e_0$ and $c$ are homotopy inverses.  Since the twisting form in the
source complex is exactly $e_0^*H$, pullback by $c$ is a quasi-isomorphism
\begin{equation}
 \begin{split}
 c^*:\quad&
 \Bigl(
 \Omega^\bullet(PM)[[u,u^{-1}]],
 d+u^{-1}e_0^*H
 \Bigr)\\
 &\longrightarrow
 \Bigl(
 \Omega^\bullet(M)[[u,u^{-1}]],
 d+u^{-1}H
 \Bigr).
 \end{split}
 \label{eq:path-to-target-quasi-isomorphism}
\end{equation}
Indeed, $e_0\circ c=\operatorname{id}_M$, and
$e_0\circ r_s=e_0$ for every $s$.  The interval homotopy operator
for $r_s$ therefore anticommutes with multiplication by $e_0^*H$.
The ordinary homotopy formula is consequently also a chain homotopy
for $d+u^{-1}e_0^*H$, proving the assertion for the twisted complex.

Combining the fixed point theorem with this path space contraction gives the
main result of the section.

\begin{theorem}
\label{thm:brane-to-target}
The chosen circle action and normalized measure determine an isomorphism
\begin{equation}
 \begin{split}
 \operatorname{Loc}:
 h^\bullet_{\mathbb T}
 (S^2M,\nabla^{\mathcal L_H}:\overline H)
 \longrightarrow
 H^\bullet\Bigl(
 \Omega^\bullet(M)[[u,u^{-1}]],
 d+u^{-1}H
 \Bigr),
 \end{split}
 \label{eq:brane-to-target-isomorphism}
\end{equation}
given by
\begin{equation}
 \operatorname{Loc}
 =c^*\circ\Psi^{-1}\circ j^*.
 \label{eq:localization-map}
\end{equation}
Since $c^*\eta=0$, this map is also restriction to constant sphere
maps in the canonical trivialization of their coefficient line.
The factorization displays the two geometric steps which establish
that restriction is an isomorphism.
The construction is summarized by the commutative diagram
\begin{equation}
\begin{CD}
h^\bullet_{\mathbb T}(S^2M,\nabla^{\mathcal L_H}:\overline H)
@>{j^*}>>
H^\bullet(PM,d+u^{-1}H_P)\\
@V{\operatorname{Loc}}V{\cong}V
@V{c^*\circ\Psi^{-1}}V{\cong}V\\
H^\bullet(M,d+u^{-1}H)
@=
H^\bullet(M,d+u^{-1}H).
\end{CD}
 \label{eq:localization-diagram}
\end{equation}
\end{theorem}

The theorem also shows which components of the brane space survive
localization.

\begin{corollary}
Let $S^2_\alpha M$ be a component corresponding to a nonzero class
$\alpha\in[S^2,M]$.  Then
\begin{equation}
 h^\bullet_{\mathbb T}
 (S^2_\alpha M,\nabla^{\mathcal L_H}:\overline H)=0.
 \label{eq:nonzero-component-vanishing}
\end{equation}
Only the component containing the constant maps contributes to
\eqref{eq:brane-to-target-isomorphism}.
\end{corollary}

\begin{proof}
Every fixed map factors through the interval $[-1,1]$ and is therefore null
homotopic.  Thus $S^2_\alpha M$ has no fixed points when $\alpha\neq0$.  The
conclusion follows from Lemma \ref{lem:no-fixed-points}.
\end{proof}

\section{The twisted brane Bismut-Chern character}
\label{sec:brane-character}

A connection on a vector bundle gives more than its curvature.  It also
gives parallel transport along paths, and hence holonomy around loops.
The Bismut-Chern character brings these two aspects of a connection
together in a differential form on the free loop space.  We shall use
this observation to construct characteristic forms on the brane
space in the presence of a gerbe.

What we want is quite concrete.  If
$\iota_M:M\to S^2M$ denotes the inclusion of constant sphere maps, the
new form should restrict under $\iota_M$ to the usual twisted Chern
form on $M$.  It should also be closed for the superconnection $Q_H$,
and a change of module connections should give an exact difference.
We shall establish these properties by constructing the form and its
transgression explicitly.  The relation with the localization theorem
will then follow from the restriction to constant maps.

There is a natural starting point.  A sphere map restricts to a loop on
each latitude, so it determines a family to which the loop-space theory
can be applied.  The difficulty is that the twisted Bismut-Chern forms
at different latitudes take values in different fibres of holonomy line
bundles.  Before these forms can be integrated, their coefficient lines must be compared.  We obtain this comparison by viewing
the same sphere as a circular family of meridian paths.  Parallel
transport along a meridian identifies the fibres at its different
latitudes.  We average the transported connections, and only then apply
the Bismut-Chern construction in the circular direction.  The essential
point is that the resulting holonomy line, with its connection, pulls
back to the average latitude holonomy line $\mathcal L_H$ already used
in the brane complex.

We begin with a brief review of the characteristic forms on the target and on its loop space.  We then construct the averaged gerbe and its
modules, define the brane forms and their transgressions, and prove the
restriction and localization identities.  Write $\mathcal G$ for the
gerbe of Section \ref{sec:brane-space}, and retain the normalized measure
of that section.  Periodic series have the meaning fixed in
\eqref{eq:brane-exotic-complex}. 

\subsection{Twisted Chern characters and their loop-space refinements}
\label{subsec:character-review}

We recall the geometric description of twisted $K$-theory in
\cite{BCMMS,MS}, with the signs and periodic normalization of
\cite{HM15}.  A gerbe module with connection consists locally of bundles
$E_\alpha\to\sU_\alpha$ and connection-preserving isomorphisms
\begin{equation}
 \phi_{\alpha\beta}:L_{\alpha\beta}\otimes E_\beta
 \longrightarrow E_\alpha,
 \label{eq:character-module-descent}
\end{equation}
compatible with the gerbe product.  Since
$F_{\alpha\beta}=B_\beta-B_\alpha$, their curvatures satisfy
\begin{equation}
 \phi_{\alpha\beta}^{-1}
 (B_\alpha\operatorname{id}+F^{E_\alpha})
 \phi_{\alpha\beta}
 =B_\beta\operatorname{id}+F^{E_\beta}.
 \label{eq:character-module-curvature-descent}
\end{equation}
Thus the combination of module curvature and gerbe curving has the
transformation law needed for a Chern-Weil construction.

For a general twist we use the relative Hilbert-module model
of \cite[Section~4]{MS}.  The trace-class reduction is part of the
module data.  In compatible paired local trivializations, the
transition operators of $E$ and $E'$ agree modulo the trace-class
ideal, and admissible module connections have a difference which is
a smooth trace-class one-form, with all derivatives taken in the
trace norm.  Equivalently, the pair has a common connection after
passing to the bounded operators modulo the trace-class ideal.
Compatibility includes the overlap module identifications, so the
comparison is preserved when a path crosses between local charts.

Only paired differences are traced.  In particular, equality of the
quotient transports makes the difference of the two monodromies
trace class; a condition on curvature alone would not suffice for
the degree-zero loop character.  The existence of compatible
relative connections is part of the model recalled in \cite{MS}.
Finite-rank representatives, when available, use ordinary traces,
and finite-rank summands are retained in stabilization.  We do not
require the individual Hilbert-module curvatures to be trace class.

The periodic twisted Chern form is given locally by
\begin{equation}
 \operatorname{Ch}_H(\nabla^E,\nabla^{E'})|_{\sU_\alpha}
 =e^{-u^{-1}B_\alpha}
 \Tr\left(e^{-u^{-1}F^{E_\alpha}}
          -e^{-u^{-1}F^{E'_\alpha}}\right).
 \label{eq:character-target-chern-form}
\end{equation}
Equation \eqref{eq:character-module-curvature-descent} shows that these
local forms agree.  The Bianchi identity and $dB_\alpha=H$ give
\begin{equation}
 (d+u^{-1}H)\operatorname{Ch}_H(\nabla^E,\nabla^{E'})=0.
 \label{eq:character-target-chern-closedness}
\end{equation}
The resulting class depends only on $[E,E']\in K^0(M,H)$.
Suspension gives the character in odd degree, with the periodic
normalization recalled below.  In this way one obtains
\[
 \operatorname{Ch}_H:K^\bullet(M,H)\longrightarrow
 H^\bullet\bigl(\Omega^\bullet(M)[[u,u^{-1}]],d+u^{-1}H\bigr).
\]

The passage to loops adds parallel transport to this curvature
description.  Let $LM=C^\infty(\mathbb T,M)$, let $K_L$ generate loop
rotation, and let $\iota_L:M\to LM$ be the constant-loop inclusion.
For an ordinary bundle $V$ with connection, Bismut's construction
\cite{B85} gives a completed even form $\operatorname{BCh}(\nabla^V)$
with
\begin{equation}
 (d-u\iota_{K_L})\operatorname{BCh}(\nabla^V)=0,
 \qquad
 \iota_L^*\operatorname{BCh}(\nabla^V)
 =\Tr e^{-u^{-1}F^V}.
 \label{eq:character-review-ordinary-bismut}
\end{equation}
Its degree-zero term is the trace of inverse holonomy in our orientation
convention.  The higher terms are circular iterated integrals of
curvature, with the curvature operators brought to a common fibre by
parallel transport.  On a constant loop the transports become
identities, and the simplex integrals become the coefficients of the
curvature exponential.

For a gerbe, the analogous form has coefficients in its holonomy line
bundle $\mathcal L_H^{\mathrm{loop}}\to LM$.  Put
\[
 \overline H_L=\int_{\mathbb T}\operatorname{ev}_\theta^*H\,d\theta,
 \qquad
 D_H^L=\nabla^{\mathcal L_H^{\mathrm{loop}}}
       -u\iota_{K_L}+u^{-1}\overline H_L,
\]
where $\operatorname{ev}_\theta:LM\to M$ is evaluation.  The twisted
Bismut-Chern form of \cite[Section~3]{HM15} satisfies
\begin{equation}
 D_H^L\operatorname{BCh}_H(\nabla^E,\nabla^{E'})=0,
 \qquad
 \iota_L^*\operatorname{BCh}_H(\nabla^E,\nabla^{E'})
 =\operatorname{Ch}_H(\nabla^E,\nabla^{E'}).
 \label{eq:character-review-twisted-bismut}
\end{equation}
The second identity uses the canonical trivialization of the holonomy
line over constant loops.  Locally, the curving exponential corrects
the transformation of the iterated integrals under tensoring a module
with $L_{\alpha\beta}$.  Their remaining holonomy factor is precisely
the transition factor of the coefficient line.

It is this last feature that matters for the sphere.  For each latitude $t$ there
is a restriction map
\[
 r_t:S^2M\longrightarrow LM,\qquad r_t(x)=x\circ c_t.
\]
The form $r_t^*\operatorname{BCh}_H$ takes values in
$r_t^*\mathcal L_H^{\mathrm{loop}}$, which varies with $t$.
The averaged transition functions defining $\mathcal L_H$ do not by
themselves identify these coefficient bundles.  The following
construction supplies a common geometric setting for the connections
before any characteristic form is averaged or traced.

\subsection{Latitude families and the averaged gerbe}

Put $a=-\pi/2$, $b=\pi/2$, and $I=[a,b]$.  We use the
\textbf{latitude path space}
\begin{equation}
 \mathscr P M=C^\infty(I,M),
 \qquad
 e_t:\mathscr P M\longrightarrow M,\quad\gamma\longmapsto\gamma(t).
 \label{eq:character-latitude-path-space}
\end{equation}
Both endpoints are free.  The parameter here is latitude, whereas the
parameter of $PM=C^\infty([-1,1],M)$ in
Section \ref{sec:brane-localization} is height.  These two uses of paths
are related by the smooth maps
\begin{equation}
 \begin{aligned}
 \kappa:PM&\longrightarrow\mathscr P M,
 & (\kappa\gamma)(t)&=\gamma(\sin t),\\
 \Gamma:S^2M&\longrightarrow L(\mathscr P M),
 & \Gamma(x)(\theta)(t)&=x(c_t(\theta)).
 \end{aligned}
 \label{eq:character-currying}
\end{equation}
For fixed $\theta$, the function $t\mapsto x(c_t(\theta))$ is the
image of a meridian.  As $\theta$ goes round the circle, these paths
form a loop in $\mathscr P M$.  This is the meaning of $\Gamma$.
It is equivariant for source rotation and loop rotation.  Its
smoothness follows from the smooth parametrization
$(\theta,t)\mapsto c_t(\theta)$ on $\mathbb T\times I$ and the
exponential law for smooth mapping spaces.

For a form $\omega$ on $\mathscr P M$, write
\[
 \mathsf A\omega=\int_{\mathbb T}\operatorname{ev}_\theta^*\omega\,d\theta,
 \qquad
 \operatorname{ev}_\theta:L(\mathscr P M)\longrightarrow\mathscr P M,
 \quad\lambda\longmapsto\lambda(\theta).
\]

\subsubsection{The averaged gerbe and its holonomy line bundle}

For an open set $\sU_\alpha$ of the brane Brylinski cover, put
$\mathscr P\sU_\alpha=C^\infty(I,\sU_\alpha)$.  These sets cover
$\mathscr P M$.  Indeed, if $\gamma\in\mathscr P M$, the sphere map
\[
 p\longmapsto\gamma\left(a+\frac{b-a}{2}(z(p)+1)\right)
\]
has the same image as $\gamma$, where $z:S^2\to[-1,1]$ is height.
The defining property of the brane cover therefore places this image
in some $\sU_\alpha$.

For loops of paths we use the invariant open cover of $L(\mathscr P M)$ by the sets $L(\mathscr P\sU_\alpha)$; its union is 
\begin{equation}
 \mathscr W=\bigcup_\alpha L(\mathscr P\sU_\alpha).
 \label{eq:character-loop-path-neighborhood}
\end{equation}
Indeed, a loop $\lambda$ in $\mathscr P M$ corresponds to a
smooth map $F_\lambda:\mathbb T\times I\to M$.  Choose a smooth
surjection $q:S^2\to\mathbb T\times I$.  The sphere map
$F_\lambda\circ q$ has the same image as $F_\lambda$, so the
defining property of the brane Brylinski cover places this image
in some $\sU_\alpha$.  Thus
$\lambda\in L(\mathscr P\sU_\alpha)$.

Define the averaged curvings and flux by
\begin{equation}
 B_\alpha^\rho=\int_a^b e_t^*B_\alpha\,\rho(t)\,dt,
 \qquad
 H_\rho=\int_a^b e_t^*H\,\rho(t)\,dt.
 \label{eq:character-path-curving}
\end{equation}
To give these forms a gerbe interpretation, we must choose connections on the overlap line bundles.  Let $L\to\sU$ be a line bundle with
connection, and let
$P_t^L:e_a^*L\to e_t^*L$ be parallel transport along a path from
$a$ to $t$.  The transported connections all act on the single
bundle $e_a^*L$, so their affine average is defined by
\begin{equation}
 \nabla^{L,\rho}
 =e_a^*\nabla^L+
 \int_a^b(\nabla_t^L-e_a^*\nabla^L)\rho(t)\,dt,
 \label{eq:character-averaged-line-connection}
\end{equation}
where
$$ \nabla_t^L=(P_t^L)^{-1}\circ e_t^*\nabla^L\circ P_t^L. $$
In particular,
$e_t^*\nabla^L$ differentiates $P_t^L$ as the path varies.
It follows that $\nabla_t^L-e_a^*\nabla^L$ is a well-defined scalar
1-form on $\mathscr P\sU$.

\begin{proposition}
\label{prop:character-path-gerbe}
The bundles $L_{\alpha\beta}^\rho=e_a^*L_{\alpha\beta}$, with the
connections \eqref{eq:character-averaged-line-connection} and the gerbe
product isomorphisms pulled back by $e_a$, form a gerbe $\mathcal G_\rho$
on $\mathscr P M$ with curvings $B_\alpha^\rho$ and curvature $H_\rho$.
Its underlying gerbe is $e_a^*\mathcal G$.
\end{proposition}

\begin{proof}
We first check that the averaged connections respect the gerbe
product.  Fix a triple intersection and write
\[
 \mu:L_{\alpha\beta}\otimes L_{\beta\gamma}
 \longrightarrow L_{\alpha\gamma}
\]
for its connection-preserving multiplication.  Over
$\mathscr P\sU_{\alpha\beta\gamma}$, put $\mu_t=e_t^*\mu$
and write $P_t^{\alpha\beta}$ for parallel transport in
$L_{\alpha\beta}$ from $a$ to $t$.  Since $\mu$ preserves
connections, parallel transport satisfies
\[
 P_t^{\alpha\gamma}\circ\mu_a
 =\mu_t\circ
 (P_t^{\alpha\beta}\otimes P_t^{\beta\gamma}).
\]
After transport to the initial fibres, the same product $\mu_a$
therefore intertwines the transported connections for every $t$:
\[
 \nabla_t^{L_{\alpha\gamma}}\circ\mu_a
 =\mu_a\circ\left(
 \nabla_t^{L_{\alpha\beta}}\otimes1
 +1\otimes\nabla_t^{L_{\beta\gamma}}\right).
\]
The map $\mu_a$ is independent of $t$, although it may vary with the
path; this dependence is included in the connection identity.

To see why averaging preserves this identity, set
\[
 C_t^{\alpha\beta}
 =\nabla_t^{L_{\alpha\beta}}-e_a^*\nabla^{L_{\alpha\beta}}.
\]
These are scalar one-forms on the path space.  Subtracting the
connection identity at $t=a$ from the one at $t$ gives
\[
 C_t^{\alpha\gamma}=C_t^{\alpha\beta}+C_t^{\beta\gamma}.
\]
The same equality holds for their averages
$C_\rho^{\alpha\beta}=\int_a^b C_t^{\alpha\beta}\rho(t)\,dt$.
Consequently the connections
$\nabla^{L_{\alpha\beta},\rho}
=e_a^*\nabla^{L_{\alpha\beta}}+C_\rho^{\alpha\beta}$
are compatible with $\mu_a$.  Associativity is unchanged, since the
product maps themselves are simply pulled back by $e_a$.

We now carry out the curvature calculation.  For a line bundle $L$,
write
\[
 \nabla_a=e_a^*\nabla^L,\qquad
 \nabla_t^L=\nabla_a+C_t,\qquad
 C_\rho=\int_a^b C_t\rho(t)\,dt,
\]
and put $F_a=e_a^*F^L$.  Since $C_t$ is scalar,
$C_t\wedge C_t=0$.  The curvature of a transported connection is
therefore
\[
 F_a+dC_t=(\nabla_t^L)^2
 =(P_t^L)^{-1}e_t^*F^L P_t^L=e_t^*F^L.
\]
The final equality uses the fact that curvature acts by scalars on a
line bundle.  Applying the same calculation to $C_\rho$ gives
\begin{align*}
& (\nabla^{L,\rho})^2\\
 =&F_a+dC_\rho\\
 =&F_a+\int_a^b(e_t^*F^L-F_a)\rho(t)\,dt\\
 =&\int_a^b e_t^*F^L\,\rho(t)\,dt.
\end{align*}
Thus the normalization of the measure cancels the initial curvature
term.  In particular,
\begin{equation}
 (\nabla^{L_{\alpha\beta},\rho})^2
 =\int_a^b e_t^*F_{\alpha\beta}\,\rho(t)\,dt
 =B_\beta^\rho-B_\alpha^\rho.
 \label{eq:character-path-overlap-curvature}
\end{equation}
Also $dB_\alpha^\rho=H_\rho$.  These are the required connection and
curving identities.  The underlying line bundles and their products
are pulled back by $e_a$, which proves the assertion about the
integral gerbe class.
\end{proof}

On $\mathscr W$, the holonomy line of $\mathcal G_\rho$ has
frames $\sigma_\alpha^\rho$ over $L(\mathscr P\sU_\alpha)$.
Its transition functions are the holonomies of
$L_{\alpha\beta}^\rho$, and its connection form is
$-\iota_{K_\ell}\mathsf A B_\alpha^\rho$, where $K_\ell$
generates loop rotation.  We denote this line bundle with connection
over $\mathscr W$ by
$(\mathscr L_\rho,\nabla^{\mathscr L_\rho})$.
Since $\Gamma$ takes values in $\mathscr W$, these local data
define its pullback on the entire brane space.

\subsubsection{Comparison with the average latitude holonomy line}
We now compare the pullback of $\mathscr L_\rho$ under $\Gamma$ with the coefficient line of the brane theory.
\begin{proposition}[Comparison of holonomy line bundles]
\label{prop:character-line-comparison}
There is an equivariant isomorphism preserving connections
\begin{equation}
 \Gamma^*(\mathscr L_\rho,\nabla^{\mathscr L_\rho})
 \cong(\mathcal L_H,\nabla^{\mathcal L_H}).
 \label{eq:character-line-comparison}
\end{equation}
It identifies $\Gamma^*\sigma_\alpha^\rho$ with $s_\alpha$.  Moreover,
\begin{equation}
 \Gamma^*\mathsf A H_\rho=\overline H,
 \qquad
 \Gamma^*D_\rho=Q_H\Gamma^*,
 \quad
 D_\rho=\nabla^{\mathscr L_\rho}
       -u\iota_{K_\ell}+u^{-1}\mathsf A H_\rho.
 \label{eq:character-superconnection-comparison}
\end{equation}
\end{proposition}

\begin{proof}
Fix $x\in S^2\sU_{\alpha\beta}$.  Along the loop $\Gamma(x)$, the bundle
$e_a^*L_{\alpha\beta}$ has the constant fibre
$(L_{\alpha\beta})_{x(S)}$.  A nonzero vector in this fibre gives a periodic
frame in the circular direction.  Transport it along the meridians to
latitude $t$, and write $q_t$ for the connection 1-form in the resulting
frame along that circle.  With our convention for parallel transport,
\[
 \operatorname{hol}_{\nabla^{L_{\alpha\beta}}}(x\circ c_t)
 =\exp\left(-\int_{\mathbb T}q_t\right).
\]
The function $-\int_{\mathbb T}q_t$ is continuous in $t$ and is zero at
$a$.  It is therefore exactly the normalized logarithm
$\ell_{\alpha\beta}(t;x)$ in
\eqref{eq:normalized-log-holonomy}.  The averaged connection on the
constant fibre has 1-form $\int_a^b q_t\rho(t)\,dt$.  Its holonomy is
thus
\begin{equation}
 \operatorname{hol}_{\nabla^{L_{\alpha\beta},\rho}}(\Gamma(x))
 =\exp\left(\int_a^b\ell_{\alpha\beta}(t;x)\rho(t)\,dt\right)
 =g_{\alpha\beta}(x).
 \label{eq:character-holonomy-comparison}
\end{equation}
This identifies the frames and the transition functions.  Fubini's theorem
gives
$\Gamma^*\mathsf A B_\alpha^\rho=\overline{B_\alpha}$ and
$\Gamma^*\mathsf A H_\rho=\overline H$.  Equivariance of $\Gamma$ then
identifies the connection forms and the contractions by the two generating
vector fields.  This proves all the assertions.
\end{proof}

The proposition explains why we average before applying the character formula.  Forms on different latitudes initially have
different coefficient lines.  The path-space gerbe supplies a single
connection whose loop holonomy has the required averaged transition law.

\subsection{Gerbe modules and averaged connections}

We now apply the same parallel transport and averaging to the admissible
gerbe modules of Section \ref{subsec:character-review}.  Compatibility
with the overlap line connections will give modules for the averaged
gerbe.  The difference between scalar and operator-valued connections shows up in the curvature, which we compute below.

Set $E_\alpha^\rho=e_a^*E_\alpha$.  If
$P_t^{E_\alpha}:e_a^*E_\alpha\to e_t^*E_\alpha$ denotes parallel transport,
define
\begin{equation}
 \begin{split}
 \nabla_t^{E_\alpha}
 &=(P_t^{E_\alpha})^{-1}\circ e_t^*\nabla^{E_\alpha}
                            \circ P_t^{E_\alpha},\\
 \nabla^{E_\alpha,\rho}
 &=e_a^*\nabla^{E_\alpha}+
   \int_a^b(\nabla_t^{E_\alpha}-e_a^*\nabla^{E_\alpha})\rho(t)\,dt.
 \end{split}
 \label{eq:character-averaged-module-connection}
\end{equation}
Define $(E'{}^\rho,\nabla^{E',\rho})$ in the same manner.

\begin{proposition}
\label{prop:character-averaged-modules}
The bundles $E_\alpha^\rho$, with the connections
\eqref{eq:character-averaged-module-connection} and the isomorphisms
$e_a^*\phi_{\alpha\beta}$, form a module for $\mathcal G_\rho$.  The
construction preserves direct sums and admissible pairs.  Along the
inclusion $c_{\mathrm{lat}}:M\to\mathscr P M$ of constant paths, it
restricts to the original module with connection.
\end{proposition}

\begin{proof}
Compatibility of \eqref{eq:character-module-descent} with connections
implies compatibility with parallel transport.  After transporting to the
initial fibres, the connection identity holds for every $t$ with the fixed
isomorphism $e_a^*\phi_{\alpha\beta}$.  Averaging gives
\begin{equation}
 (e_a^*\phi_{\alpha\beta})^{-1}\nabla^{E_\alpha,\rho}
 (e_a^*\phi_{\alpha\beta})
 =\nabla^{L_{\alpha\beta},\rho}\otimes1
  +1\otimes\nabla^{E_\beta,\rho}.
 \label{eq:character-averaged-module-descent}
\end{equation}
This proves the module assertion and also proves compatibility with direct
sums.

To check admissibility, choose compatible local charts for the Hilbert
bundles.  Variation of parameters expresses the difference between the
two parallel transports as an integral with one factor equal to the
difference of the connection forms.  The compatible quotient identifications also compare the transports
across changes of chart.  The trace-class operators form a two-sided
ideal in the bounded operators, and the integral is over a compact
interval.  The difference is therefore trace class.  The same
argument applies after differentiation in the path variables and shows
that the difference of the transported, and hence of the averaged,
connections is smooth in the trace norm.  This verifies admissibility.
On constant paths the parallel transports are identities and all the
pulled back connections coincide.  Normalization of the measure proves
the last assertion.
\end{proof}

It is useful to record the curvature of the averaged connection without
suppressing its noncommutative term.  Write
$\nabla_t^{E_\alpha}=\nabla_a^{E_\alpha}+C_t$ and
$C_\rho=\int_a^b C_t\rho(t)\,dt$.  Then
\begin{equation}
 F^{E_\alpha,\rho}
 =\int_a^b(P_t^{E_\alpha})^{-1}e_t^*F^{E_\alpha}
                   P_t^{E_\alpha}\,\rho(t)\,dt
 +C_\rho\wedge C_\rho
   -\int_a^b C_t\wedge C_t\,\rho(t)\,dt.
 \label{eq:character-averaged-module-curvature}
\end{equation}
This follows by expanding the curvature of
$\nabla_a^{E_\alpha}+C_\rho$.  It reduces to the curvature average for a
line bundle, as used above.  For a module, the full expression
\eqref{eq:character-averaged-module-curvature} is the curvature which enters
the character.  Thus the meridian transport contributes to the circular
iterated integrals even before the latter are performed.

\subsection{The twisted brane Bismut-Chern form}

We recall the form of the loop construction needed here, fixing its
orientation convention.  Let $V$ be a bundle with connection on an open
set of $\mathscr P M$.  On its loop space, let
$U_s^V:\operatorname{ev}_0^*V\to\operatorname{ev}_s^*V$ be parallel
transport in the circular direction, and put
\[
 R_V(s)=(U_s^V)^{-1}(\operatorname{ev}_s^*F^V) U_s^V.
\]
For $\Delta_n=\{0\leq s_1\leq\cdots\leq s_n\leq1\}$, define the
endomorphism-valued series
\begin{equation}
 \mathcal W(V)=
 \left(1+\sum_{n\geq1}(-u)^{-n}
 \int_{\Delta_n}R_V(s_1)\cdots R_V(s_n)\,ds_1\cdots ds_n\right)
 (U_1^V)^{-1}.
 \label{eq:character-loop-iterated-series}
\end{equation}
The inverse monodromy agrees with the convention $K_x=dx(V)$ and the
differential $d-u\iota_K$ used in this paper.

For a line bundle $L$ it
gives
\begin{equation}
 \mathcal W(L)=\operatorname{hol}_{\nabla^L}^{-1}
                e^{-u^{-1}\mathsf A F^L}.
 \label{eq:character-line-iterated-series}
\end{equation}

On $L(\mathscr P\sU_\alpha)$ set
\begin{equation}
 \begin{split}
 \omega_\alpha^\rho
 &=\Tr\left(\mathcal W(E_\alpha^\rho)
            -\mathcal W((E'_\alpha)^\rho)\right),\\
 \beta_\alpha^\rho
 &=e^{-u^{-1}\mathsf A B_\alpha^\rho}
       \omega_\alpha^\rho\otimes\sigma_\alpha^\rho.
 \end{split}
 \label{eq:character-path-loop-local-form}
\end{equation}
The connections used in the first line are the averaged connections of
\eqref{eq:character-averaged-module-connection}.

\begin{lemma}
\label{lem:character-loop-calculus}
The expressions \eqref{eq:character-path-loop-local-form} are smooth
completed differential forms.  They are invariant under loop rotation,
agree on intersections, and define an $\mathscr L_\rho$-valued
form $\beta^\rho$ on $\mathscr W$ satisfying
$D_\rho\beta^\rho=0$.
\end{lemma}

\begin{proof}
We explain why the loop calculation applies on the part of the path space
under consideration.  Each coefficient of
\eqref{eq:character-loop-iterated-series} contains finitely many curvature
insertions and an integral over a compact simplex.  Parallel transport
depends smoothly on the path and on smooth families of paths.  This follows
from its ordinary differential equation; in Hilbert-bundle charts it is
an equation in bounded operators.  For a pair of connections as above,
differences of transports and curvatures are smooth in the trace norm.
Expanding a difference of products as a sum with one such factor shows
that every paired coefficient is trace class.  This includes the term with no curvature insertions, which is
the paired difference of inverse monodromies.  The operator-ideal
estimate
\[
 \|ATB\|_1\leq\|A\|\,\|T\|_1\,\|B\|
\]
bounds every term by the trace norm of its trace-class factor and
the operator norms of the remaining factors.  These bounds hold on compact
parameter families and persist after differentiation.  They justify the
integrations and traces coefficient by coefficient.  No analytic
convergence of the formal series in $u^{-1}$ is required.

The Bismut identity in these conventions is
\begin{equation}
 (d-u\iota_{K_\ell})\omega_\alpha^\rho=0.
 \label{eq:character-local-bismut-identity}
\end{equation}
It follows from the loop calculation recalled in
Section \ref{subsec:character-review}; see \cite{B85,HM15}.
Differentiating parallel
transport inserts curvature; the Bianchi identity treats the interior
terms, and Stokes' formula on $\Delta_n$ treats the faces.  The internal
faces cancel in pairs.  The end faces match the contraction term in the
next coefficient, while cyclicity of the trace identifies the two ends
of the loop.  These are identities for a connection evaluated on a
finite collection of tangent vectors and smooth parameter families.
The preceding smoothness and trace estimates justify the identical
calculation here.  In the relative case, commutators reduce to
commutators with a trace-class factor, whose traces vanish.  Squaring
$d-u\iota_{K_\ell}$ in \eqref{eq:character-local-bismut-identity}
also proves rotation invariance.

Let $h_{\alpha\beta}^\rho$ be the holonomy of
$L_{\alpha\beta}^\rho$.  The tensor product rule for parallel transport,
together with \eqref{eq:character-line-iterated-series}, gives
\[
 \omega_\alpha^\rho
 =(h_{\alpha\beta}^\rho)^{-1}
 e^{-u^{-1}\mathsf A(B_\beta^\rho-B_\alpha^\rho)}
 \omega_\beta^\rho.
\]
The paired trace is unchanged by a compatible change of relative
trivializations: the extra terms are commutators with a trace-class
factor and have zero trace.  Since
$\sigma_\alpha^\rho=h_{\alpha\beta}^\rho\sigma_\beta^\rho$,
the forms $\beta_\alpha^\rho$ agree on intersections.

Finally, in these
frames,
\[
 D_\rho=e^{-u^{-1}\mathsf A B_\alpha^\rho}
          (d-u\iota_{K_\ell})
         e^{u^{-1}\mathsf A B_\alpha^\rho}.
\]
Equation \eqref{eq:character-local-bismut-identity} proves the required
closedness.
\end{proof}

\begin{definition}
The {\bf twisted brane Bismut-Chern form} of the admissible pair
$(E,E')$ with connection is
\begin{equation}
 \operatorname{BBCh}_H(\nabla^E,\nabla^{E'})
 =\Gamma^*\beta^\rho
 \in\Omega^{\mathrm{even}}(S^2M,\mathcal L_H)^{\mathbb T}
 [[u,u^{-1}]],
 \label{eq:character-brane-definition}
\end{equation}
where coefficients are identified by
\eqref{eq:character-line-comparison}.  Equivalently, its local formula is
\begin{equation}
 \operatorname{BBCh}_H|_{S^2\sU_\alpha}
 =e^{-u^{-1}\overline{B_\alpha}}
   \Gamma^*\omega_\alpha^\rho\otimes s_\alpha.
 \label{eq:character-brane-local-form}
\end{equation}
\end{definition}

To study the dependence on connections, it is useful to make the
transgression form part of the construction. Fix an admissible pair
$(E,E')$ and a smooth path of admissible module connections
$\nabla_s=(\nabla_s^E,\nabla_s^{E'})$, $0\leq s\leq1$.
Let $p_M:M\times[0,1]\to M$ be projection, and equip the pulled
back modules with the connections
$\widetilde\nabla^E=\nabla_s^E+ds\,\partial_s$ and
$\widetilde\nabla^{E'}=\nabla_s^{E'}+ds\,\partial_s$.
The gerbe on the product is $p_M^*\mathcal G$.  There is a natural
map
\[
 F_I:S^2M\times[0,1]\longrightarrow S^2(M\times[0,1]),
 \qquad F_I(x,s)(p)=(x(p),s).
\]
Apply the brane construction to these product connections and put
\begin{equation}
 \widetilde\beta
 =F_I^*\operatorname{BBCh}_{p_M^*H}
       (\widetilde\nabla^E,\widetilde\nabla^{E'}).
 \label{eq:character-transgression-cylinder}
\end{equation}
Its coefficient line and twisting form on $S^2M\times[0,1]$ are
pulled back from $\mathcal L_H$ and $\overline H$.  Its restriction
at parameter $s$ is $\operatorname{BBCh}_H(\nabla_s)$.

\begin{definition}[\bf Twisted brane Bismut-Chern-Simons form]
\label{def:character-brane-cs}
\leavevmode\par\noindent
The twisted brane Bismut-Chern-Simons form associated to the path
$\nabla_\bullet$ is the odd completed form
\begin{equation}
 \operatorname{BBCS}_H(\nabla_\bullet)
 =\int_0^1\iota_{\partial_s}\widetilde\beta\,ds
 \in\Omega^{\mathrm{odd}}(S^2M,\mathcal L_H)^{\mathbb T}
 [[u,u^{-1}]].
 \label{eq:character-brane-cs-form}
\end{equation}
For the affine path
$\nabla_s=(1-s)\nabla_0+s\nabla_1$, we also write
$\operatorname{BBCS}_H(\nabla_0,\nabla_1)$.
\end{definition}

The definition records the chosen path of connections.  The affine
convention gives a specified form for a pair of endpoints; the
transgression identity itself holds for every smooth admissible path.
All integrals are taken coefficientwise, and the trace estimates used
in Lemma \ref{lem:character-loop-calculus} apply to this compact
parameter family.

\begin{theorem}
\label{thm:character-even}
The form \eqref{eq:character-brane-definition} is $Q_H$-closed.
For a smooth admissible path of module connections, its variation is
given by
\begin{equation}
 \operatorname{BBCh}_H(\nabla_1)-\operatorname{BBCh}_H(\nabla_0)
 =Q_H\operatorname{BBCS}_H(\nabla_\bullet).
 \label{eq:character-brane-transgression}
\end{equation}
Consequently, its cohomology class depends only on the class of
$(E,E')$ in $K^0(M,H)$ and defines an additive homomorphism
\begin{equation}
 \operatorname{BBCh}_H:K^0(M,H)
 \longrightarrow h^{\mathrm{even}}_{\mathbb T}
 (S^2M,\nabla^{\mathcal L_H}:\overline H).
 \label{eq:character-even-homomorphism}
\end{equation}
\end{theorem}

\begin{proof}
Closedness follows from Lemma \ref{lem:character-loop-calculus} and
Proposition \ref{prop:character-line-comparison}.  The same argument
on the product shows that \eqref{eq:character-transgression-cylinder}
is closed for the product differential.  Write
\[
 \widetilde\beta=\beta_s+ds\wedge\gamma_s,
 \qquad \beta_s=\operatorname{BBCh}_H(\nabla_s).
\]
Separating the two components of the closedness identity gives
\[
 Q_H\beta_s=0,\qquad \partial_s\beta_s=Q_H\gamma_s.
\]
Integration in $s$ proves
\eqref{eq:character-brane-transgression}.

Direct sums give sums of relative traces.  An isomorphism of modules
preserves the form when their connections are transported with it,
and the transgression identity removes the dependence on this choice
of connections.  Finally, addition of the same module with the same
connection to both entries contributes zero.  These are the
isomorphism and stabilization relations in the gerbe-module description
of $K^0(M,H)$, so the form defines the stated homomorphism.
\end{proof}

\subsection{Localization and the twisted brane Bismut-Chern character}

Let $\iota_{\mathscr P}:\mathscr P M\to L(\mathscr P M)$ denote
the constant-loop inclusion, and let
$\iota_M=j\circ c:M\to S^2M$ denote the inclusion of constant sphere
maps.  The maps introduced above satisfy
\begin{equation}
 \begin{CD}
 PM @>{j}>> S^2M\\
 @V{\kappa}VV @VV{\Gamma}V\\
 \mathscr P M @>{\iota_{\mathscr P}}>> \mathscr W.
 \end{CD}
 \label{eq:character-fixed-point-square}
\end{equation}
Here $\kappa^*H_\rho=H_P$ and $\kappa\circ c=c_{\mathrm{lat}}$.

\begin{proposition}
\label{prop:character-restriction}
In the canonical trivialization of the coefficient line on the fixed point
set, one has
\begin{equation}
 \begin{split}
 j^*\operatorname{BBCh}_H(\nabla^E,\nabla^{E'})
 &=\kappa^*\operatorname{Ch}_{H_\rho}
       (\nabla^{E,\rho},\nabla^{E',\rho}),\\
 \iota_M^*\operatorname{BBCh}_H(\nabla^E,\nabla^{E'})
 &=\operatorname{Ch}_H(\nabla^E,\nabla^{E'}).
 \end{split}
 \label{eq:character-restriction-identities}
\end{equation}
For a smooth admissible path of connections, one also has
\begin{equation}
 \iota_M^*\operatorname{BBCS}_H(\nabla_\bullet)
 =\operatorname{CS}_H(\nabla_\bullet),
 \label{eq:character-cs-restriction}
\end{equation}
where the target twisted Chern-Simons form is
\[
 \operatorname{CS}_H(\nabla_\bullet)
 =\int_0^1\iota_{\partial_s}
   \operatorname{Ch}_{p_M^*H}
   (\widetilde\nabla^E,\widetilde\nabla^{E'})\,ds.
\]
Consequently,
\begin{equation}
 \operatorname{Loc}\,[\operatorname{BBCh}_H(E,E')]
 =[\operatorname{Ch}_H(E,E')].
 \label{eq:character-even-localization}
\end{equation}
\end{proposition}

\begin{proof}
On a constant loop, circular parallel transport is the identity and
$R_V(s)$ is independent of $s$.  Since the volume of $\Delta_n$ is
$1/n!$, \eqref{eq:character-loop-iterated-series} restricts to
$e^{-u^{-1}F^V}$.  The coefficient-line frame restricts to $1$.
The square \eqref{eq:character-fixed-point-square} proves the first
identity.  Proposition \ref{prop:character-averaged-modules} and the
normalization of $\rho$ prove the second.  Applying this same
restriction identity on $M\times[0,1]$ and then integrating in $s$
proves \eqref{eq:character-cs-restriction}.  Finally,
$\operatorname{Loc}=c^*e^{u^{-1}\eta}j^*$ and $c^*\eta=0$, because the
homotopy defining $\eta$ is stationary on constant paths.  This gives
\eqref{eq:character-even-localization}.
\end{proof}

We also need the odd construction, since T-duality reverses parity.
Let $q:M\times S^1\to M$.  Suspension identifies $K^1(M,H)$ with
the kernel of restriction
$K^0(M\times S^1,q^*H)\to K^0(M,H)$ at a chosen point of $S^1$.
Write $\sigma a$ for the suspended class of $a\in K^1(M,H)$, and choose
a gerbe-module representative for it.  With
\[
 \jmath:S^2M\times S^1\longrightarrow S^2(M\times S^1),
 \qquad \jmath(x,s)(p)=(x(p),s),
\]
define an odd form by
\begin{equation}
 \operatorname{BBCh}_H^{\mathrm{odd}}(a)
 =u\int_{S^1}\jmath^*\operatorname{BBCh}_{q^*H}(\sigma a).
 \label{eq:character-odd-suspension}
\end{equation}
The suspension circle has total length one and the orientation used in
the ordinary odd Chern character.  Its integral is taken coefficientwise.
The factor $u$ is required by our periodic convention: a term of degree
$2k$ has coefficient $u^{-k}$, whereas after integration the degree
$2k-1$ term has coefficient $u^{-(k-1)}$.

\begin{theorem}[Twisted brane Bismut-Chern character and localization]
\label{thm:character-localization}
The even construction and \eqref{eq:character-odd-suspension} define a
natural homomorphism in both parities
\begin{equation}
 \operatorname{BBCh}_H:K^\bullet(M,H)
 \longrightarrow h^\bullet_{\mathbb T}
 (S^2M,\nabla^{\mathcal L_H}:\overline H).
 \label{eq:character-all-parities}
\end{equation}
It is represented by $Q_H$-closed forms, and the diagram
\begin{equation}
 \begin{tikzcd}[column sep=large,row sep=large]
 K^\bullet(M,H)
 \arrow[r,"\operatorname{BBCh}_H"]
 \arrow[dr,"\operatorname{Ch}_H"'] &
 h^\bullet_{\mathbb T}(S^2M,\nabla^{\mathcal L_H}:\overline H)
 \arrow[d,"\operatorname{Loc}","\cong"']\\
 &H^\bullet(\Omega^\bullet(M)[[u,u^{-1}]],d+u^{-1}H)
 \end{tikzcd}
 \label{eq:character-localization-square}
\end{equation}
commutes.  Naturality is with respect to smooth pullback of the target
and of its gerbe with connection.
\end{theorem}

\begin{proof}
The even assertion is Theorem \ref{thm:character-even} and Proposition
\ref{prop:character-restriction}.  On $S^2M\times S^1$ the coefficient
line and the twisting form are pulled back from $S^2M$.  Fibre integration
over the closed circle anticommutes with the odd product differential,
so \eqref{eq:character-odd-suspension} is closed and takes exact forms
to exact forms.  The even theorem therefore makes its class independent
of the representative of $\sigma a$.  Restricting to constant spheres
and integrating gives
\[
 u\int_{S^1}\operatorname{Ch}_{q^*H}(\sigma a)
 =\operatorname{Ch}_H^{\mathrm{odd}}(a)
\]
in periodic twisted cohomology.  This is the usual compatibility of the
Chern character with suspension, with the periodic factor explained
above.  Proposition \ref{prop:character-restriction} now proves the
odd part of the diagram as well.  Evaluation, parallel transport, the
averages, and the iterated integrals all commute with pullback, which
proves naturality.
\end{proof}

\begin{remark}
(i) The equality
$\operatorname{BBCh}_H=\operatorname{Loc}^{-1}\operatorname{Ch}_H$
characterizes the resulting homomorphism on cohomology.  The construction
above supplies its differential-form representatives and the transgression
between different connections.  The curvature terms in
\eqref{eq:character-averaged-module-curvature} and the circular parallel
transports describe the geometry of these representatives on the brane
space.

(ii) The construction also respects the levels which occur in T-duality.
Replacing $\mathcal G$ by $\mathcal G^{\otimes m}$, with the dual gerbe
used for negative $m$, gives
\begin{equation}
 \operatorname{BBCh}_{mH}:K^\bullet(M,mH)
 \longrightarrow h^\bullet_{\mathbb T}
 (S^2M,\nabla^{\mathcal L_H^{\otimes m}}:m\overline H),
 \qquad m\in\mathbb Z.
 \label{eq:character-levels}
\end{equation}
At every level the localization identity is
$\operatorname{Loc}_m\operatorname{BBCh}_{mH}=\operatorname{Ch}_{mH}$.
Thus these characters assemble over the level variable $y$.  We now turn
to the Hori morphisms and to the relation between the level-one character
and twisted $K$-theory T-duality.
\end{remark}

\section{Graded genus-0 T-duality with \texorpdfstring{$H$}{H}-flux on brane spaces}
\label{sec:genus-zero-t-duality}

The localization theorem of Section \ref{sec:brane-localization} brings the exotic twisted
cohomology of the brane space back down to the periodic twisted cohomology of
the target. This raises a further question: if two targets are related by T-duality, does the relation persist, in a suitable sense, on their brane spaces? The localization isomorphism provides such a relation on
cohomology.  It carries a brane class to the target, where the
ordinary Hori transform applies, and the localization isomorphism
on the dual side then gives the corresponding brane class.

We first develop the cohomological argument. We begin by recalling the geometry
of T-duality with $H$-flux, together with the classical Hori isomorphism. We
then introduce its level $m$ refinement and its completed periodic forms.
With this finite-dimensional preparation in hand, we return to the average
holonomy line bundle and establish an invariant brane localization theorem at
each level. Finally, we transport the periodic Hori morphisms through brane
localization; the resulting maps assemble into a single graded genus-0 Hori morphism, whose composite with its dual is governed by the Euler operator in the level variable.  We then use the character constructed in
Section \ref{sec:brane-character} to relate the level-one morphism to
twisted $K$-theory T-duality.

\subsection{Review of T-duality with \texorpdfstring{$H$}{H}-flux}

At its simplest, T-duality replaces a circle of radius $R$ by a circle of
radius $R^{-1}$.  In the presence of background flux, the circle fibration and the
flux must be considered together.  The original flux determines the
topology of the dual circle bundle, while the topology of the original bundle
returns as part of the dual flux.  We recall the mathematical formulation of
this principle following \cite{BEM04a,BEM04b}, and use the notation of
\cite{HM24,HM25}.  It is useful to settle one convention at the outset.  The
circle rotating the source sphere will be denoted by $\mathbb T_{\rm rot}$,
whereas $\mathbb T$ and $\widehat{\mathbb T}$ will always denote the structure
circles of the two spacetime bundles.

The geometric data are displayed in the diagram
\[
\begin{tikzcd}[column sep=large, row sep=large]
\mathbb{T}
  \arrow[r]
&
(Z,A,H)
  \arrow[dr, "\pi"']
&&
(\widehat{Z},\widehat{A},\widehat{H})
  \arrow[dl, "\widehat{\pi}"]
&
\widehat{\mathbb{T}}
  \arrow[l]
\\
&&
X.
\end{tikzcd}
\]
Here $A$ and $\widehat A$ are invariant connection 1-forms, normalized to
have integral one on the corresponding circle fibres.  The forms
$H\in\Omega^3(Z)$ and $\widehat H\in\Omega^3(\widehat Z)$ are closed and
invariant, and have integral periods.  We suppress factors of
$2\pi\sqrt{-1}$.   $F_A$ and $F_{\widehat A}$ denote the curvature forms on
$X$,
\begin{equation}
 dA=\pi^*F_A,
 \qquad
 d\widehat A=\widehat\pi^*F_{\widehat A}.
 \label{eq:T-dual-curvatures}
\end{equation}

The two total spaces meet naturally over $X$ in the correspondence space
\begin{equation}
 Z\times_X\widehat Z
 =\{(z,\widehat z)\in Z\times\widehat Z:
       \pi(z)=\widehat\pi(\widehat z)\}.
 \label{eq:correspondence-space}
\end{equation}
This space is also known as the doubled space.  If $p$ and $\widehat p$ denote
the two projections, its geometry is summarized by the commutative diagram
\begin{center}
\begin{tikzcd}[column sep=large, row sep=large]
 &Z \times_X \widehat{Z} \arrow[dl, "p"'] \arrow[rd, "\widehat{p}"] & \\
Z\arrow[rd, "\pi"'] & & \widehat{Z} \arrow[ld, "\widehat{\pi}"] \\
 & X. &
\end{tikzcd}
\end{center}

The pairs $(Z,[H])$ and $(\widehat Z,[\widehat H])$ are said to be T-dual over
$X$ when
\begin{equation}
 [F_{\widehat A}]=\pi_*[H],
 \qquad
 [F_A]=\widehat\pi_*[\widehat H]
 \quad\text{in }H^2(X,\mathbb Z),
 \label{eq:T-dual-pushforward-relations}
\end{equation}
and
\begin{equation}
 p^*[H]=\widehat p^*[\widehat H]
 \quad\text{in }H^3(Z\times_X\widehat Z,\mathbb Z).
 \label{eq:topological-T-duality-correspondence}
\end{equation}
Thus T-duality exchanges the first Chern class of either circle bundle with the
fibre integral of the flux on the other.

The existence of the dual pair is already visible in the Gysin sequence.  For
$\pi:Z\to X$ the relevant portion is
\begin{center}
\begin{tikzcd}
\cdots \arrow[r]
& H^{3}(X,\mathbb{Z}) \arrow[r, "\pi^{*}"]
& H^3(Z,\mathbb{Z}) \arrow[r, "\pi_{*}"]
& H^{2}(X,\mathbb{Z}) \arrow[r, "\lbrack F_A\rbrack\wedge"]
& H^{4}(X,\mathbb{Z}) \arrow[r]
& \cdots .
\end{tikzcd}
\end{center}
Starting with $[H]\in H^3(Z,\mathbb Z)$, set
$[F_{\widehat A}]=\pi_*[H]$ and choose a principal circle bundle
$\widehat\pi:\widehat Z\to X$ with this first Chern class.  The Gysin sequence
for $\widehat Z$ is
\begin{center}
\begin{tikzcd}
\cdots \arrow[r]
& H^{3}(X,\mathbb{Z}) \arrow[r, "\widehat{\pi}^{*}"]
& H^3(\widehat{Z},\mathbb{Z}) \arrow[r, "\widehat{\pi}_{*}"]
& H^{2}(X,\mathbb{Z}) \arrow[r, "\lbrack F_{\widehat A}\rbrack\wedge"]
& H^{4}(X,\mathbb{Z}) \arrow[r]
& \cdots .
\end{tikzcd}
\end{center}
Exactness of the first sequence gives
\[
 [F_A]\wedge[F_{\widehat A}]
 =[F_A]\wedge\pi_*[H]=0.
\]
Exactness of the second therefore provides a class
$[\widehat H]\in H^3(\widehat Z,\mathbb Z)$ satisfying
$[F_A]=\widehat\pi_*[\widehat H]$.  The global statement, which may be regarded
as a geometric form of the Buscher rules, is the following theorem.

\begin{theorem}[\cite{BEM04a}]
\label{thm:existence-T-dual-pair}
Let $\pi:Z\to X$ be a principal circle bundle and let
$[H]\in H^3(Z,\mathbb Z)$.  There is a principal circle bundle
$\widehat\pi:\widehat Z\to X$ and a class
$[\widehat H]\in H^3(\widehat Z,\mathbb Z)$ such that
$(Z,[H])$ and $(\widehat Z,[\widehat H])$ are T-dual.  The resulting pair is
unique up to bundle automorphism.  In particular, the construction is
involutive.
\end{theorem}

The same construction admits a useful local description, which will also fix
our conventions for the curvings.  Let
$(H,B_\alpha,F_{\alpha\beta},L_{\alpha\beta})$ be a gerbe with connection on
$Z$, and choose a connection one-form $A$ on $Z$.  If $v$ is the vector field
generating the $\mathbb T$-action, define
\begin{equation}
 \widehat A_\alpha=-\iota_vB_\alpha.
 \label{newconn}
\end{equation}
On $U_\alpha\times\widehat{\mathbb T}$, with fibre coordinate
$\widehat\theta_\alpha$, the form
\[
 \widehat A=\widehat A_\alpha+d\widehat\theta_\alpha
\]
is a connection one-form.  The descent relations for the gerbe data are exactly
what is needed to glue these local expressions.  In this way one recovers the
T-dual circle bundle $\widehat Z\to X$ together with its connection.

We may assume, without loss of generality, that $H$ is
$\mathbb T$-invariant.  Consider
\[
 \Omega=H-A\wedge F_{\widehat A}.
\]
Since $\iota_v\Omega=0$ and $\mathcal L_v\Omega=0$, the form $\Omega$ is basic
and hence descends to a three-form on $X$.  Suppressing pullbacks from the base,
define
\begin{equation}
 \widehat H=F_A\wedge\widehat A+\Omega.
 \label{split}
\end{equation}
This is a closed three-form on $\widehat Z$.  On the correspondence space one
then obtains the differential form version of the Buscher relation,
\begin{equation}
 \widehat p^*\widehat H
 =p^*H+d\bigl(p^*A\wedge\widehat p^*\widehat A\bigr).
 \label{eq:T-dual-flux-Buscher-relation}
\end{equation}

The curving on the dual side can be read from the same identity.  On
$U_\alpha\times\mathbb T\times\widehat{\mathbb T}$, the Buscher rules give
\begin{equation}
 \widehat B_\alpha
 =B_\alpha+A\wedge\widehat A
   -d\theta_\alpha\wedge d\widehat\theta_\alpha.
 \label{buscher}
\end{equation}
Indeed,
\[
\begin{split}
 \iota_v\widehat B_\alpha
 &=\iota_v\bigl(
 B_\alpha+A\wedge\widehat A
 -d\theta_\alpha\wedge d\widehat\theta_\alpha
 \bigr)\\
 &=-\widehat A_\alpha+\widehat A-d\widehat\theta_\alpha=0.
\end{split}
\]
Thus $\widehat B_\alpha$ is a two-form on $\widehat Z$, rather than merely a
form on the correspondence space, and $d\widehat B_\alpha=\widehat H$.  The
remaining descent equations produce the full T-dual gerbe with connection
$(\widehat H,\widehat B_\alpha,\widehat F_{\alpha\beta},
\widehat L_{\alpha\beta})$; see \cite{HM24,HM25}.

It is convenient to collect the preceding relation into the two-form
\begin{equation}
 \mathcal F=-p^*A\wedge\widehat p^*\widehat A.
 \label{eq:T-duality-two-form}
\end{equation}
Equation \eqref{eq:T-dual-flux-Buscher-relation} is then equivalent to
\begin{equation}
 p^*H-\widehat p^*\widehat H=d\mathcal F.
 \label{eq:T-duality-flux-identity}
\end{equation}
This simple transgression identity is the differential form heart of the Hori
transform.  Write $d_H=d+H\wedge$ and
$d_{\widehat H}=d+\widehat H\wedge$.

\begin{theorem}[Hori isomorphism \cite{BEM04a}]
\label{thm:classical-Hori-isomorphism}
The map
\begin{equation}
 \begin{split}
 T:(\Omega^\bullet(Z)^{\mathbb T},d_H)&\longrightarrow
 (\Omega^{\bullet+1}(\widehat Z)^{\widehat{\mathbb T}},-d_{\widehat H}),\\
 \omega&\longmapsto
 \widehat p_*\bigl(e^{\mathcal F}p^*\omega\bigr)
 =\int_{\mathbb T}
 e^{-p^*A\wedge\widehat p^*\widehat A}p^*\omega
 \end{split}
 \label{eq:classical-Hori-map}
\end{equation}
is an isomorphism of $\mathbb Z_2$-graded complexes.  It induces the
parity-reversing isomorphism
\begin{equation}
 T:H^\bullet(Z,H)\xrightarrow{\ \cong\ }
 H^{\bullet+1}(\widehat Z,\widehat H).
 \label{eq:classical-Hori-cohomology}
\end{equation}
\end{theorem}

This formulation contains both the topological and the metric pictures.  If a
metric $g_X$ is fixed on the base, the customary metrics on the total spaces are
\[
 g_Z=\pi^*g_X+R^2A\mathbin\odot A,
 \qquad
 g_{\widehat Z}=\widehat\pi^*g_X
 +R^{-2}\widehat A\mathbin\odot\widehat A.
\]
The radius is therefore exchanged with its reciprocal, while
\eqref{eq:T-dual-pushforward-relations} exchanges topology with flux.

When local data are required, we choose a good cover $\{U_\alpha\}$ of $X$
which trivializes both circle bundles.  The induced cover of the correspondence
space consists of the sets
$U_\alpha\times\mathbb T\times\widehat{\mathbb T}$.  We reserve the ordinary
letters $U_\alpha$ for these base charts, and retain $\sU,\sV$ for the open sets
on general manifolds and mapping spaces used in the preceding sections.

The classical Hori transform treats one unit of flux.  The brane-space theory,
however, naturally involves all tensor powers of the average holonomy line
bundle.  This leads to an integral level $m$, and hence to the graded periodic
form of the Hori morphism which we now recall.

\subsection{Periodic twisted cohomology and graded Hori morphisms}

In this subsection we review the graded Hori map and the corresponding form of
T-duality developed in \cite{HM21}.  For $m\in\mathbb Z$, put
\begin{equation}
 d_{mH}=d+mH\wedge,
 \qquad
 d_{m\widehat H}=d+m\widehat H\wedge.
 \label{eq:level-m-twisted-differentials}
\end{equation}
The level $m$ Hori map is
\begin{equation}
 T_m:\Omega^\bullet(Z)^{\mathbb T}\longrightarrow
 \Omega^{\bullet+1}(\widehat Z)^{\widehat{\mathbb T}},
 \qquad
 T_m(\omega)=\widehat p_*
 \left(e^{m\mathcal F}p^*\omega\right).
 \label{eq:level-m-Hori-map}
\end{equation}
Here and below, the displayed degree shift is understood modulo two.  The invariant
subcomplex computes the twisted cohomology of the total space
\cite{CG,HM18}.  We shall therefore identify its cohomology with
$H^\bullet(Z,mH)$, and make the analogous identification on the dual side.

\begin{proposition}
\label{prop:level-m-Hori}
With the fibre orientations induced by the circle actions, one has
\begin{equation}
 d_{m\widehat H}T_m=-T_m d_{mH}.
 \label{eq:level-m-Hori-chain-relation}
\end{equation}
Consequently, the Hori map induces
\begin{equation}
 T_m:H^\bullet(Z,mH)\longrightarrow
 H^{\bullet+1}(\widehat Z,m\widehat H).
 \label{eq:level-m-Hori-cohomology}
\end{equation}
There is a dual Hori map $\widehat T_m$ satisfying
\begin{equation}
 \widehat T_m\circ T_m=-m\operatorname{id},
 \qquad
 T_m\circ\widehat T_m=-m\operatorname{id}.
 \label{eq:level-m-Hori-compositions}
\end{equation}
In particular, $T_m$ is an isomorphism  whenever $m\neq0$.
\end{proposition}

\begin{proof}
Equation \eqref{eq:T-duality-flux-identity} gives
\begin{equation}
 \left(d+m\widehat p^*\widehat H\right)
 \left(e^{m\mathcal F}p^*\omega\right)
 =e^{m\mathcal F}p^*(d+mH)\omega.
\end{equation}
Integration over an oriented one-dimensional fibre anticommutes with the odd
differential, and this proves
\eqref{eq:level-m-Hori-chain-relation}.

To compute the compositions, recall that every invariant form has a unique
decomposition
\begin{equation}
 \omega=\pi^*\omega_0+A\wedge\pi^*\omega_1.
 \label{eq:invariant-form-decomposition}
\end{equation}
Using $e^{m\mathcal F}=1-mA\wedge\widehat A$ in
\eqref{eq:level-m-Hori-map}, one obtains
\begin{equation}
 T_m(\omega)=\widehat\pi^*\omega_1
 -m\widehat A\wedge\widehat\pi^*\omega_0.
 \label{eq:explicit-level-m-Hori}
\end{equation}
The same calculation on the dual side, with the reversed correspondence form
and the compatible fibre orientation, proves
\eqref{eq:level-m-Hori-compositions}.
\end{proof}

At level one this recovers the usual Hori isomorphism
\begin{equation}
 T_1:H^\bullet(Z,H)\xrightarrow{\ \cong\ }
 H^{\bullet+1}(\widehat Z,\widehat H).
 \label{eq:usual-Hori-isomorphism}
\end{equation}
For the localization argument, however, it is better to place this construction
in the same completed periodic setting as the exotic theory of the brane space.

Let
\begin{equation}
 R=\mathbb C[u,u^{-1}],
 \qquad |u|=2,
\end{equation}
and use the form-degree completion specified in
Section~\ref{subsec:brane-periodic-completion}.  On the finite
dimensional targets this adds no further completion at fixed total
degree.  For an
invariant flux $H$, define
\begin{equation}
 D^u_{mH}=d+u^{-1}mH\wedge
 \label{eq:periodic-level-m-differential}
\end{equation}
on $\Omega^\bullet(Z)^{\mathbb T}\widehat\otimes R$.  The appearance of $u$ requires a
minor rescaling before the ordinary Hori map can be used.  This is the same
rescaling that appears in exotic twisted cohomology of loop spaces
\cite{HM15}.

\begin{lemma}[Periodic rescaling]
\label{lem:periodic-rescaling}
For a homogeneous form $\omega_j\in\Omega^j(Z)$, define
\begin{equation}
 c_u(\omega_j)=u^{-\lfloor j/2\rfloor}\omega_j
 \label{eq:periodic-rescaling-map}
\end{equation}
and extend $R$-linearly.  The map $c_u$ induces an isomorphism
\begin{equation}
 c_{u,mH}:H^\bullet
 \left(\Omega^\bullet(Z)^{\mathbb T}\widehat\otimes R,d_{mH}\right)
 \xrightarrow{\ \cong\ }
 H^\bullet
 \left(\Omega^\bullet(Z)^{\mathbb T}\widehat\otimes R,D^u_{mH}\right).
 \label{eq:periodic-rescaling-isomorphism}
\end{equation}
\end{lemma}

\begin{proof}
If $\omega$ has parity $k\in\{0,1\}$, a direct calculation gives
\begin{equation}
 c_u d_{mH}\omega
 =u^{-k}D^u_{mH}c_u\omega.
 \label{eq:periodic-rescaling-intertwining}
\end{equation}
The factor $u^{-k}$ is invertible.  Hence $c_u$ takes cocycles to cocycles and
coboundaries to coboundaries.  Its inverse is obtained by reversing the powers
of $u$ in \eqref{eq:periodic-rescaling-map}.
\end{proof}

We may now define the periodic level $m$ Hori morphism on cohomology by
\begin{equation}
 T_{m,u}=c_{u,m\widehat H}\circ T_m\circ c_{u,mH}^{-1}.
 \label{eq:periodic-level-m-Hori}
\end{equation}
The dual morphism $\widehat T_{m,u}$ is defined in the same way.  By
construction,
\begin{equation}
 \widehat T_{m,u}\circ T_{m,u}=-m\operatorname{id},
 \qquad
 T_{m,u}\circ\widehat T_{m,u}=-m\operatorname{id}.
 \label{eq:periodic-Hori-compositions}
\end{equation}

To retain all levels at once, put
\begin{equation}
 \mathcal P_m(Z,H)=
 H^\bullet\left(
 \Omega^\bullet(Z)^{\mathbb T}\widehat\otimes R,D^u_{mH}
 \right)
 \label{eq:periodic-target-level}
\end{equation}
and define the graded periodic theory
\begin{equation}
 \mathcal P(Z,H)=
 \bigoplus_{m\in\mathbb Z}\mathcal P_m(Z,H)y^m.
 \label{eq:graded-periodic-target-theory}
\end{equation}
The variable $y$ records the level.  The maps $T_{m,u}$ assemble into
\begin{equation}
 \operatorname{GHor}_u:\mathcal P(Z,H)\longrightarrow
 \mathcal P(\widehat Z,\widehat H),
 \qquad
 \operatorname{GHor}_u\left(\sum_m a_my^m\right)
 =\sum_mT_{m,u}(a_m)y^m.
 \label{eq:graded-periodic-Hori}
\end{equation}
Together with the dual morphism, equation
\eqref{eq:periodic-Hori-compositions} becomes
\begin{equation}
 \widehat{\operatorname{GHor}}_u\circ\operatorname{GHor}_u
 =-y\frac{\partial}{\partial y},
 \qquad
 \operatorname{GHor}_u\circ\widehat{\operatorname{GHor}}_u
 =-y\frac{\partial}{\partial y}.
 \label{eq:graded-periodic-Euler-relation}
\end{equation}

This completes the finite dimensional part of the argument.  The next step is
to show that the localization theorem remains valid after one imposes the
structure-circle symmetry and replaces the original gerbe by any of its tensor
powers.

\subsection{Level \texorpdfstring{$m$}{m} exotic theory and invariant localization}

Assume that the gerbes with connection on $Z$ and $\widehat Z$ are equivariant
for the respective fibre-circle actions, and that their curvature forms are the
invariant forms $H$ and $\widehat H$ chosen above.  Postcomposition induces
actions of $\mathbb T$ and $\widehat{\mathbb T}$ on the corresponding brane
spaces.  These actions commute with the action of $\mathbb T_{\rm rot}$, which
rotates the source sphere.

Equations \eqref{eq:normalized-log-holonomy} and
\eqref{eq:average-transition} show that the average holonomy construction is
compatible with tensor products of gerbes.  Consequently, the average holonomy
line bundle associated to the level $m$ gerbe is
\begin{equation}
 \mathcal L_{mH}=\mathcal L_H^{\otimes m}\longrightarrow S^2Z.
 \label{eq:level-m-average-line-bundle}
\end{equation}
For $m<0$, the right-hand side denotes the corresponding tensor power of the
dual line bundle, and for $m=0$ it denotes the trivial line bundle.  The
curvature formula \eqref{eq:average-curvature} scales linearly with $m$ and gives
\begin{equation}
 \left(\nabla^{\mathcal L_H^{\otimes m}}\right)^2
 =m\iota_K\overline H.
 \label{eq:level-m-average-curvature}
\end{equation}
Here $K$ is the vector field generated by $\mathbb T_{\rm rot}$.

\begin{definition}
For $m\in\mathbb Z$, define
\begin{equation}
 Q_{mH}=\nabla^{\mathcal L_H^{\otimes m}}
 -u\iota_K+u^{-1}m\overline H
 \label{eq:level-m-brane-superconnection}
\end{equation}
and put
\begin{equation}
 \mathcal E_m(Z,H)=H^\bullet\left(
 \Omega^\bullet(S^2Z,\mathcal L_H^{\otimes m})^{
 \mathbb T_{\rm rot}\times\mathbb T}\widehat\otimes R,
 Q_{mH}
 \right).
 \label{eq:level-m-brane-exotic-theory}
\end{equation}
\end{definition}

The cancellation which made the original superconnection equivariantly flat
persists at every level.

\begin{proposition}
The operator $Q_{mH}$ satisfies
\begin{equation}
 Q_{mH}^2+u\mathcal L_K^{\mathcal L_H^{\otimes m}}=0.
 \label{eq:level-m-equivariant-flatness}
\end{equation}
It therefore defines the cohomology in
\eqref{eq:level-m-brane-exotic-theory}.
\end{proposition}

\begin{proof}
By \eqref{eq:level-m-average-curvature}, the curvature term in the square is
$m\iota_K\overline H$.  The commutator of $-u\iota_K$ with
$u^{-1}m\overline H$ is $-m\iota_K\overline H$, so the two terms cancel.  By
Cartan's formula, the remaining terms give
$-u\mathcal L_K^{\mathcal L_H^{\otimes m}}$.
\end{proof}

We can now invoke the localization theorem of Section 2.  What has to be checked is that its construction respects the extra action of the structure circle.  This equivariance is the reason for having kept
$\mathbb T_{\rm rot}$ distinct from $\mathbb T$ throughout.

\begin{theorem}[Level $m$ localization]
\label{thm:invariant-level-m-localization}
For every $m\in\mathbb Z$, Theorem \ref{thm:brane-to-target} induces an
isomorphism
\begin{equation}
 \operatorname{Loc}_{Z,m}:\mathcal E_m(Z,H)
 \xrightarrow{\ \cong\ }\mathcal P_m(Z,H).
 \label{eq:invariant-level-m-localization}
\end{equation}
The analogous map
\begin{equation}
 \operatorname{Loc}_{\widehat Z,m}:
 \mathcal E_m(\widehat Z,\widehat H)
 \xrightarrow{\ \cong\ }
 \mathcal P_m(\widehat Z,\widehat H)
 \label{eq:dual-invariant-level-m-localization}
\end{equation}
is also an isomorphism.
\end{theorem}

\begin{proof}
Apply Theorem \ref{thm:brane-to-target} to the level $m$ gerbe.  Its average
holonomy line bundle, connection, and averaged flux are given by
\eqref{eq:level-m-average-line-bundle},
\eqref{eq:level-m-average-curvature}, and $m\overline H$, respectively.  The
resulting localization map is
\begin{equation}
 \operatorname{Loc}_{Z,m}=c^*\circ\Psi_m^{-1}\circ j^*,
 \label{eq:level-m-localization-map}
\end{equation}
where $\Psi_m$ is multiplication by the level $m$ form of the gauge factor in
\eqref{eq:path-gauge-isomorphism}.

The maps $j$ and $c$, as well as the contraction of the path space, are
equivariant for the postcomposition action of $\mathbb T$.  The forms $H_P$ and
$\eta$ which enter the construction are invariant because $H$ is invariant.
Choose the auxiliary metric and partitions of unity invariant
under the compact product of the structure circle and the rotation
circle.  The normal contraction, interval transports and contracting
operators in the localization proof then respect both actions.
The Mayer-Vietoris argument and the path contraction consequently
apply within the $\mathbb T$-invariant subcomplex.  Equivalently,
whenever an invariant cocycle has a primitive in that argument,
averaging the primitive over the structure circle supplies an
invariant one.  This proves
\eqref{eq:invariant-level-m-localization}, including the
quasi-isomorphism assertion after passage to invariant forms.  The argument on the dual side is
identical.
\end{proof}

\subsection{Graded Hori morphisms on brane spaces}

For each $m\in\mathbb Z$, define the level $m$ genus-0 Hori morphism by
transporting the periodic Hori morphism through brane localization:
\begin{equation}
 T^{S^2}_m=
 \operatorname{Loc}_{\widehat Z,m}^{-1}\circ
 T_{m,u}\circ\operatorname{Loc}_{Z,m}:
 \mathcal E_m(Z,H)\longrightarrow
 \mathcal E_m(\widehat Z,\widehat H).
 \label{eq:level-m-brane-Hori}
\end{equation}
The dual morphism $\widehat T^{S^2}_m$ is defined in the same
manner.  These are morphisms of cohomology groups, obtained from the
finite dimensional transform through the two localization isomorphisms.

To assemble the levels, define
\begin{equation}
 \mathcal E_{S^2}(Z,H)=
 \bigoplus_{m\in\mathbb Z}\mathcal E_m(Z,H)y^m
 \label{eq:graded-brane-exotic-theory}
\end{equation}
and define $\mathcal E_{S^2}(\widehat Z,\widehat H)$ similarly.  The maps in
\eqref{eq:level-m-brane-Hori} assemble into the
\textbf{graded genus-0 Hori morphism}
\begin{equation}
 \operatorname{GHor}_{S^2}:\mathcal E_{S^2}(Z,H)
 \longrightarrow\mathcal E_{S^2}(\widehat Z,\widehat H),
 \label{eq:graded-brane-Hori}
\end{equation}
given by
\begin{equation}
 \operatorname{GHor}_{S^2}\left(\sum_m\alpha_my^m\right)
 =\sum_mT^{S^2}_m(\alpha_m)y^m.
 \label{eq:graded-brane-Hori-formula}
\end{equation}

Let
\begin{equation}
 \operatorname{Loc}_Z=\bigoplus_{m\in\mathbb Z}
 \operatorname{Loc}_{Z,m},
 \qquad
 \operatorname{Loc}_{\widehat Z}=\bigoplus_{m\in\mathbb Z}
 \operatorname{Loc}_{\widehat Z,m}.
\end{equation}
The construction is summarized by the commutative diagram
\begin{equation}
\begin{CD}
\mathcal E_{S^2}(Z,H)
@>{\operatorname{GHor}_{S^2}}>>
\mathcal E_{S^2}(\widehat Z,\widehat H)\\
@V{\operatorname{Loc}_Z}V{\cong}V
@V{\operatorname{Loc}_{\widehat Z}}V{\cong}V\\
\mathcal P(Z,H)
@>{\operatorname{GHor}_u}>>
\mathcal P(\widehat Z,\widehat H).
\end{CD}
 \label{eq:graded-brane-Hori-diagram}
\end{equation}
The diagram makes the argument transparent.  Every structural identity on the
bottom row is carried to the corresponding identity on the top row.

\begin{theorem}[Graded genus-0 T-duality with $H$-flux]
\label{thm:graded-genus-zero-T-duality}
The graded genus-0 Hori morphisms satisfy
\begin{equation}
 \widehat{\operatorname{GHor}}_{S^2}\circ
 \operatorname{GHor}_{S^2}
 =-y\frac{\partial}{\partial y},
 \qquad
 \operatorname{GHor}_{S^2}\circ
 \widehat{\operatorname{GHor}}_{S^2}
 =-y\frac{\partial}{\partial y}.
 \label{eq:graded-brane-Euler-relation}
\end{equation}
Consequently, the graded genus-0 Hori morphism is an isomorphism on the
direct sum of the nonzero levels.
\end{theorem}

\begin{proof}
On the level $m$ summand, equations \eqref{eq:level-m-brane-Hori} and
\eqref{eq:periodic-Hori-compositions} give
\begin{equation}
 \widehat T^{S^2}_m\circ T^{S^2}_m=-m\operatorname{id},
 \qquad
 T^{S^2}_m\circ\widehat T^{S^2}_m=-m\operatorname{id}.
\end{equation}
Multiplication by $m$ on the level $m$ summand is precisely the action of
$y\frac{\partial}{\partial y}$.  This proves
\eqref{eq:graded-brane-Euler-relation}.  When $m\neq0$, the inverse of
$T^{S^2}_m$ is $-m^{-1}\widehat T^{S^2}_m$.
\end{proof}

The level-1 statement is the direct genus-0 counterpart of the classical
Hori isomorphism.

\begin{corollary}[Genus-0 T-duality with $H$-flux]
\label{cor:brane-exotic-T-duality}
At level 1, there is an odd isomorphism
\begin{equation}
 T^{S^2}_1:\mathcal E_1(Z,H)
 \xrightarrow{\ \cong\ }
 \mathcal E_1(\widehat Z,\widehat H).
 \label{eq:brane-exotic-T-duality}
\end{equation}
Equivalently, T-duality exchanges the two parities of the completed periodic
exotic twisted cohomology of the brane spaces.
\end{corollary}

\begin{proof}
This is the level-1 part of Theorem
\ref{thm:graded-genus-zero-T-duality}.  Its inverse is
$-\widehat T^{S^2}_1$.
\end{proof}

\begin{remark}
The level-0 summand is annihilated by the Euler operator and is therefore not
part of the preceding isomorphism statement.  The full graded Hori morphism is
an isomorphism after restriction to the nonzero levels, or, equivalently, after
inverting $y\frac{\partial}{\partial y}$.
\end{remark}

\begin{remark}
The morphism in \eqref{eq:level-m-brane-Hori} is defined on cohomology.  A direct Hori integral at the level of differential forms would
instead call for an integration procedure over the fibres of the
mapping-space correspondence
\begin{equation}
 S^2(Z\times_X\widehat Z)
 \cong S^2Z\times_{S^2X}S^2\widehat Z,
\end{equation}
and these fibres are infinite dimensional.  Localization avoids this analytic
difficulty by reducing the construction to the ordinary Hori transform.  The
resulting morphism depends on the circle action and on the normalized latitude
measure chosen in Sections 1 and 2.
\end{remark}

\subsection{Compatibility with the twisted brane Bismut-Chern character}
\label{subsec:character-t-duality}

The Hori morphisms constructed above act on the whole exotic cohomology. The character of Section \ref{sec:brane-character} singles out a distinguished family of classes, namely those coming from gerbe modules on the target.  We now
identify the action of the level-1 Hori morphism on these classes.

Choose the gerbe identification on the correspondence space and the
orientation of the circle pushforward as in the T-duality data.  The
classical twisted $K$-theory isomorphism is
\begin{equation}
 T_K:K^j(Z,H)\longrightarrow K^{j+1}(\widehat Z,\widehat H),
 \qquad j\in\mathbb Z/2\mathbb Z.
 \label{eq:character-k-theory-t-duality}
\end{equation}
It is obtained by pullback, the identification of twists on the
correspondence space, and the oriented circle pushforward
\cite{BEM04a,BunkeSchick,HM15}.  The fibre has trivial vertical tangent
bundle, so its characteristic factor in the Chern character pushforward
formula is $1$.  With the same signs as in \eqref{eq:classical-Hori-map},
the classical compatibility and its periodic version read
\begin{equation}
 \operatorname{Ch}_{\widehat H}T_K
 =T_{1,u}\operatorname{Ch}_H.
 \label{eq:character-periodic-target-t-duality}
\end{equation}
Here the Chern characters are periodic: on forms of degree $r$ the
rescaling is $u^{-\lfloor r/2\rfloor}$, as in
\eqref{eq:periodic-rescaling-map}.  This convention agrees with both
\eqref{eq:character-target-chern-form} and the factor $u$ in the odd
suspension formula \eqref{eq:character-odd-suspension}.

To use the invariant complexes of this section, average the brane
representative over the structure circle.  The coefficient-line action
and connection are those used in
\eqref{eq:level-m-brane-exotic-theory}.  They commute with $Q_H$, and
hence averaging defines
\begin{equation}
 \operatorname{BBCh}^{\mathbb T}_H:
 K^\bullet(Z,H)\longrightarrow\mathcal E_1(Z,H),
 \qquad
 a\longmapsto
 \left[\int_{\mathbb T}g^*\operatorname{BBCh}_H(a)\,dg\right].
 \label{eq:character-invariant-representative}
\end{equation}
There is no need to choose an equivariant gerbe module to form this
average.  On the target, inclusion of invariant forms computes the full
twisted cohomology.  Indeed, filter the complex by ordinary form degree.
The filtration is finite, and comparison reduces to ordinary de Rham
cohomology, where averaging over the connected circle is homotopic to
the identity.  The same argument applies with periodic coefficients.
Accordingly, we write $\operatorname{Ch}_H(a)\in\mathcal P_1(Z,H)$
for the class of the averaged target representative.  Since localization
commutes with the structure-circle action, Theorem
\ref{thm:character-localization} gives
\begin{equation}
 \operatorname{Loc}_{Z,1}\operatorname{BBCh}^{\mathbb T}_H
 =\operatorname{Ch}_H.
 \label{eq:character-invariant-localization}
\end{equation}
We use the analogous conventions on the dual side.

\begin{theorem}[T-duality and the twisted brane character]
\label{thm:character-t-duality}
The level-one genus-0 Hori morphism is compatible with the twisted
brane Bismut-Chern character:
\begin{equation}
 T^{S^2}_1\operatorname{BBCh}^{\mathbb T}_H
 =\operatorname{BBCh}^{\widehat{\mathbb T}}_{\widehat H}T_K.
 \label{eq:character-brane-t-duality}
\end{equation}
More precisely, the following diagram commutes:
\begin{equation}
 \begin{tikzcd}[column sep=5em,row sep=3.3em]
 K^j(Z,H)
 \arrow[r,"T_K","\cong"']
 \arrow[d,"\operatorname{BBCh}^{\mathbb T}_H"']
 \arrow[dd,out=180,in=180,looseness=1.5,
        "\operatorname{Ch}_H"']
 &K^{j+1}(\widehat Z,\widehat H)
 \arrow[d,"\operatorname{BBCh}^{\widehat{\mathbb T}}_{\widehat H}"]
 \arrow[dd,out=0,in=0,looseness=1.5,
        "\operatorname{Ch}_{\widehat H}"]\\
 \mathcal E_1^j(Z,H)
 \arrow[r,"T^{S^2}_1","\cong"']
 \arrow[d,"\operatorname{Loc}_{Z,1}"',"\cong"]
 &\mathcal E_1^{j+1}(\widehat Z,\widehat H)
 \arrow[d,"\operatorname{Loc}_{\widehat Z,1}","\cong"']\\
 \mathcal P_1^j(Z,H)
 \arrow[r,"T_{1,u}","\cong"']
 &\mathcal P_1^{j+1}(\widehat Z,\widehat H).
 \end{tikzcd}
 \label{eq:character-t-duality-diagram}
\end{equation}
The superscripts on $\mathcal E_1$ and $\mathcal P_1$ indicate
parity.  The outer curved arrows are the twisted Chern characters;
each is the composite of the two vertical arrows on its side.
\end{theorem}

\begin{proof}
The lower square is the definition of $T^{S^2}_1$.  For
$a\in K^j(Z,H)$, equations
\eqref{eq:character-invariant-localization} and
\eqref{eq:character-periodic-target-t-duality} give
\begin{align*}
 &\operatorname{Loc}_{\widehat Z,1}
 T^{S^2}_1\operatorname{BBCh}^{\mathbb T}_H(a)\\
 =&T_{1,u}\operatorname{Ch}_H(a)\\
 =&\operatorname{Ch}_{\widehat H}(T_Ka)\\
 =&\operatorname{Loc}_{\widehat Z,1}
   \operatorname{BBCh}^{\widehat{\mathbb T}}_{\widehat H}(T_Ka).
\end{align*}
The map $\operatorname{Loc}_{\widehat Z,1}$ is an isomorphism, so the
upper square commutes as well.  Equation
\eqref{eq:character-invariant-localization} on the two targets identifies
the outer curved arrows with the vertical composites, completing the
enhanced diagram.
\end{proof}

The theorem concerns cohomology classes.  The characteristic forms and
their transgressions are explicit, whereas the Hori morphism on brane
spaces remains the cohomological morphism defined by localization.
At the level of twisted $K$-theory the statement concerns the original
T-dual pair.  Although \eqref{eq:character-levels} defines a character
for every $m$, the relation
\[
 \pi_*(m[H])=m\,c_1(\widehat Z)
\]
shows that the classical dual circle bundle for $(Z,mH)$ need not be
$\widehat Z$.  Thus the characters at all levels and the graded Hori
identities coexist with the level-one $K$-theory compatibility proved
above, without requiring a $K$-theory T-duality between the same total
spaces at every level.\\

\noindent{\bf Data Availability}\\

Data sharing is not applicable to this article as no new data were created or analyzed in this study.

\end{document}